\documentclass[letterpaper,11pt]{article}
\usepackage{graphicx}
\usepackage{amssymb,amsmath,amsthm}
\usepackage{fullpage}
\usepackage{color}
\usepackage[hypertex]{hyperref}
\newcommand{\remove}[1]{}
\usepackage{xspace}
\usepackage{multirow}
\usepackage[ruled,vlined,linesnumbered]{algorithm2e}

\newcommand{\polylg}{polylogarithmic\xspace}
\newcommand{\N}{{\mathbb{N}}}

\DeclareMathSymbol{\qedsymb} {\mathord}{AMSa}{"04}
\newcommand{\eps}{\varepsilon}
\def\veps{\varepsilon}
\def\D{\mathcal D}
\def\E{\mathcal E}
\def\S{\mathcal S}

\def\eqdef{\stackrel{\text{\tiny\rm def}}{=}}
\def\sprod{\circledast}

\newcommand{\maxx}[1]{ \max\{#1\} }
\newcommand{\floor}[1]{ \lfloor{#1}\rfloor }

\newcommand{\myparagraph}[1]{\medskip\noindent{\bf #1.}}

\newtheorem{definition}{Definition}[section]
\newtheorem{lemma}[definition]{Lemma}
\newtheorem{claim}[definition]{Claim}
\newtheorem{fact}[definition]{Fact}
\newtheorem{theorem}[definition]{Theorem}
\newtheorem{corollary}[definition]{Corollary}

\DeclareMathOperator{\polylog}{polylog}
\DeclareMathOperator{\lcs}{LCS}

\DeclareMathOperator{\ed}{ed}
\DeclareMathOperator{\edd}{\underline{ed}}
\DeclareMathOperator{\blk}{blk}

\DeclareMathOperator{\EX}{\mathbb E}
\DeclareMathOperator{\cost}{cost}
\DeclareMathOperator{\Ham}{H}
\DeclareMathOperator{\supp}{supp}

\def\DTEP{{\mathrm{DTEP}}}

\newcommand{\EC}{{\cal E}}
\newcommand{\W}{{\cal W}}
\newcommand{\EE}[2]{{\mathbb{E}_{#1}\left[#2\right]}}
 
\def\compactify{\itemsep=0pt \topsep=0pt \partopsep=0pt \parsep=0pt}

\newcommand{\zo}{\{0,1\}}
\newcommand{\tO}{{\tilde{O}}}
\DeclareMathOperator{\dx}{dx}
\DeclareMathOperator{\IH}{IH}

\begin{document}

\title{
Polylogarithmic Approximation for Edit Distance\\and the Asymmetric
Query Complexity
}
\author{
Alexandr Andoni\thanks{Supported in part by NSF CCF 0832797.}\\
 Princeton University/C.C.I.
\and
Robert Krauthgamer\thanks{Supported in part by The Israel Science Foundation (grant \#452/08), and by a Minerva grant.}\\
Weizmann Institute
\and
Krzysztof Onak\thanks{Supported in part by NSF grants 0732334 and 0728645.}\\
MIT}

\maketitle

\begin{abstract}
We present a near-linear time algorithm that approximates the edit distance 
between two strings within a polylogarithmic factor; 
specifically, for strings of length $n$ and every fixed $\eps>0$, 
it can compute a $(\log n)^{O(1/\eps)}$ approximation in $n^{1+\eps}$ time. 
This is an {\em exponential} improvement over the previously known factor, 
$2^{\tilde O(\sqrt{\log n})}$, with a comparable running time~\cite{OR-edit,AO-edit}.  
Previously, no efficient polylogarithmic approximation algorithm was known
for any computational task involving edit distance 
(e.g., nearest neighbor search or sketching).

This result arises naturally in the study of 
a new \emph{asymmetric query} model.
In this model, the input consists of two strings $x$ and $y$,
and an algorithm can access $y$ in an
unrestricted manner, while being charged for querying every symbol of $x$.
Indeed, we obtain our main result by designing an algorithm that
makes a small number of queries in this model.
We then provide a nearly-matching lower bound on the 
number of queries.

Our lower bound is the first to expose hardness of edit distance
stemming from the input strings being ``repetitive'', 
which means that many of their substrings are approximately identical.
Consequently, our lower bound provides the first rigorous separation
between edit distance and Ulam distance, 
which is edit distance on non-repetitive strings, such as permutations. 
\end{abstract}

\newpage

\section{Introduction}

Manipulation of strings has long been central to computer science, 
arising from the high demand to process texts 
and other sequences efficiently.
For example, for the simple task of {\em comparing} two strings
(sequences), one of the first methods  
emerged to be the \emph{edit distance} (aka the Levenshtein distance) 
\cite{Lev65},
defined as the minimum number of character insertions, deletions, and substitutions
needed to transform one string into the other. 
This basic distance measure, together with its more elaborate versions, is widely
used in a variety of areas such as
computational biology, speech recognition, and information retrieval.
Consequently, improvements in edit distance algorithms have
the potential of major impact. As a result, computational problems
involving edit distance have been studied 
extensively (see \cite{Navarro01,Gus-book} and references therein).

The most basic problem is that of computing the edit distance between two
strings 
of length $n$ over some alphabet. 
It can be solved in $O(n^2)$ time by a classical algorithm \cite{WF74};
in fact this is a prototypical dynamic programming algorithm,
see, e.g., the textbook \cite{CLRS} and references therein. 
Despite significant research over more than three decades,
this running time has so far been improved
only slightly to $O(n^2/\log^2 n)$ \cite{MP80},
which remains the fastest algorithm known to date.%
\footnote{The result of \cite{MP80} applies to
  constant-size alphabets. It was recently
extended to arbitrarily large alphabets,
albeit with an $O(\log\log n)^2$
  factor loss in runtime
\cite{BFC08}.} 

Still, a near-quadratic runtime is often unacceptable in 
modern applications that must deal with massive datasets, such as the genomic data.
Hence practitioners tend to rely on faster heuristics~\cite{Gus-book,Navarro01}. This has motivated the quest for  faster algorithms at the expense of approximation, see, e.g.,
\cite[Section 6]{I-survey} and \cite[Section 8.3.2]{IMCRC}.
Indeed, the past decade has seen a serious effort in this direction.%
\footnote{We shall not attempt to present a complete list of results for 
restricted settings (e.g., average-case/smoothed analysis, weakly-repetitive
strings, and bounded distance-regime), for variants of the distance function
(e.g., allowing more edit operations), or for related computational problems 
(such as pattern matching, nearest neighbor search, and
sketching). See also the surveys of~\cite{Navarro01} and~\cite{Sah-Encyclopedia}.
}
One general approach 
is to design linear 
time algorithms that approximate the edit distance.
A linear-time $\sqrt{n}$-approximation algorithm immediately follows
from the exact algorithm of \cite{LMS98},
which runs in time $O(n+d^2)$, 
where $d$ is the edit distance between the input strings. 
Subsequent research improved the approximation factor, first to $n^{3/7}$ \cite{BJKK04},
then to $n^{1/3+o(1)}$ \cite{BES06},
and finally to $2^{\tO(\sqrt{\log n})}$ \cite{AO-edit}
(building on \cite{OR-edit}).
Predating some of this work was the {\em sublinear-time} algorithm of  \cite{BEK+03} achieving 
$n^\eps$  approximation, 
but only when the edit distance $d$ is rather large.

Better progress has been obtained on {\em variants} of edit distance,
where one either restricts the input strings, 
or allows additional edit operations.
An example from the first category is the edit distance on
non-repetitive strings (e.g., permutations of $[n]$), termed {\em the Ulam
  distance} in the literature.
The classical Patience Sorting algorithm computes the exact Ulam distance between
two strings in $O(n\log n)$ time.
An example in the second category is the case of two variants of
the edit distance where certain block operations are allowed. Both of these variants
admit an $\tilde O(\log
n)$ approximation in near-linear time \cite{CPSV, MSah, CMu, C-PhD}. 

Despite the efforts, achieving a polylogarithmic approximation
factor for the classical edit distance 
has eluded researchers for a long time. In fact, this is has been the case 
not only in the context of linear-time algorithms, 
but also in the related tasks, such as nearest neighbor search, 
$\ell_1$-embedding, or sketching. 
From a lower bounds perspective, only a {\em sublogarithmic}
approximation has been ruled out
for the latter two tasks \cite{KN,KR06,AK07},
thus giving evidence that a sublogarithmic approximation for the
distance computation might be much harder or even impossible to attain.

\subsection{Results}

Our first and main result is an algorithm that runs in near-linear time
and approximates edit distance within a {\em \polylg factor}.
Note that this is {\em exponentially better} than 
the previously known factor $2^{\tO(\sqrt{\log n})}$ 
(in comparable running time), due to \cite{OR-edit,AO-edit}.

\begin{theorem}[Main]
\label{thm:introAlgorithm}
For every fixed $\eps>0$, there is an algorithm that 
approximates the edit distance between two input strings $x,y\in\Sigma^n$ 
within a factor of $(\log n)^{O(1/\eps)}$, and runs in $n^{1+\eps}$ time. 
\end{theorem}

This development stems from
a principled study of edit distance in a computational model
that we call the {\em asymmetric query} model, and which we shall define shortly.
Specifically, we design a query-efficient procedure in the said model,
and then show how this procedure yields a near-linear time algorithm.
We also provide a query complexity lower bound for this model, 
which matches or nearly-matches the performance of our procedure.

A conceptual contribution of our query complexity lower bound is 
that it is the first one to expose hardness stemming from ``repetitive substrings'', 
which means that many small substrings of a string may be approximately equal.
Empirically, it is well-recognized that such repetitiveness is a 
major obstacle for designing efficient algorithms.
All previous lower bounds (in any computational model) failed to exploit it,
while in our proof the strings' repetitive structure is readily apparent.
More formally, our lower bound provides the first rigorous separation
of edit distance from Ulam distance (edit distance on non-repetitive strings). 
Such a separation was not previously known in any studied model of
computation, and in fact all the lower bounds known for the edit
distance hold to (almost) the same degree for the Ulam distance.
These models include: non-embeddability into normed spaces~\cite{KN, KR06, AK07},
lower bounds on sketching complexity~\cite{AK07, AJP-sketch}, and
(symmetric) query complexity~\cite{BEK+03, AN-ulam}. 

\myparagraph{Asymmetric Query Complexity}
Before stating the results formally, we define the problem and the
model precisely. Consider two strings $x,y\in \Sigma^n$ for some alphabet $\Sigma$,
and let $\ed(x,y)$ denote the edit distance between these two strings.
The computational problem is the promise problem known as the
Distance Threshold Estimation Problem (DTEP)~\cite{SS02}: 
distinguish whether $\ed(x,y)> R$ or $\ed(x,y)\le R/\alpha$, 
where $R>0$ is a parameter (known to the algorithm) 
and $\alpha\ge1$ is the {\em approximation factor}.
We use $\DTEP_\beta$ to denote the case of $R=n/\beta$, 
where $\beta\ge 1$ may be a function of $n$.

In the {\em asymmetric query model},
the algorithm knows in advance (has unrestricted access to) one of the strings, 
say $y$, and has only {\em query access} to the other string, $x$. 
The \emph{asymmetric query complexity} of an algorithm is the number of coordinates in $x$ 
that the algorithm has to probe in order to solve DTEP with success probability
at least $2/3$.

We now give complete statements of our upper and lower bound results.
Both exhibit a smooth {\em tradeoff} between approximation factor 
and query complexity.
For simplicity, we state the bounds in two extreme 
regimes of approximation ($\alpha=\polylog(n)$ and $\alpha=\mathrm{poly}(n)$).
See Theorem~\ref{thm:upperBound} for the full statement of the upper
bound, and Theorems~\ref{thm:main_lb} and~\ref{thm:lbMorePrecise} for
the full statement of the lower bound.

\begin{theorem}[Query complexity upper bound]
  \label{thm:introUpperBound}
For every $\beta=\beta(n)\ge 2$ and fixed $0<\eps<1$ there is an algorithm 
that solves $\DTEP_{\beta}$ with approximation $\alpha=(\log n)^{O(1/\eps)}$, 
and makes $\beta n^{\eps}$ asymmetric queries. 
This algorithm runs in time $O(n^{1+\eps})$.

For every $\beta=O(1)$ and fixed integer $t\ge 2$ there is an algorithm 
for $\DTEP_\beta$ achieving approximation $\alpha=O(n^{1/t})$,
with $O(\log^{t-1} n)$ queries into $x$. 
\end{theorem}

It is an easy observation that our general edit distance algorithm in
Theorem~\ref{thm:introAlgorithm} follows immediately from the above query
complexity upper bound theorem, by running the latter for all $\beta$
that are a power of 2. 

\begin{theorem}[Query complexity lower bound]
\label{thm:LB}
For a sufficiently large constant $\beta>1$,
every algorithm that solves $\DTEP_\beta$ 
with approximation $\alpha=\alpha(n)>2$ 
has asymmetric query complexity 
$2^{\Omega\left(\frac{\log n}{\log \alpha + \log \log n}\right)}$. 
Moreover, for every fixed non-integer $t>1$,
every algorithm that solves $\DTEP_\beta$ with approximation $\alpha=n^{1/t}$ 
has asymmetric query complexity $\Omega(\log^{\floor{t}} n)$.
\end{theorem}

We summarize in Table \ref{tbl:results} our results and previous bounds 
for $\DTEP_\beta$ under edit distance and Ulam distance. 
For completeness, we also present known results for a common
query model where the algorithm has query 
access to both strings (henceforth referred to as the \emph{symmetric query} model).
We point out two implications of our bounds on the asymmetric query complexity:
\begin{itemize}
\compactify
\item 
There is a strong separation between edit distance and Ulam distances.
In the Ulam metric, a {\em constant} approximation is achievable
with only $O(\log n)$ asymmetric queries (see \cite{ACCL04}, which builds on \cite{EKKRV00}).
In contrast, for edit distance, we show an exponentially higher
complexity lower bound,
of $2^{\Omega(\log n/\log\log n)}$, even for a larger (polylogarithmic) approximation.
\item 
Our query complexity upper and lower bounds are nearly-matching, 
at least for a range of parameters. At one extreme,
approximation $O(n^{1/2})$ can be achieved with $O(\log n)$ queries, 
whereas approximation $n^{1/2-\eps}$ already requires $\Omega(\log^2 n)$ queries. 
At the other extreme, approximation $\alpha=(\log n)^{1/\eps}$
can be achieved using $n^{O(\eps)}$ queries, 
and requires $n^{\Omega(\eps/\log\log n)}$ queries.
\end{itemize}

\begin{table}[htb]
  \centering
  \begin{tabular}[t]{|l|l|l|l|l|}
    \hline
    Model & Metric & Approx. & Complexity & Remarks \\
    \hline\hline
    \multirow{2}{*}{
      \begin{tabular}{@{}l@{}}
        Near-linear\\
        time
      \end{tabular}
    } 
    & Edit & $(\log n)^{O(1/\eps)}$ & $n^{1+\eps}$ 
      & Theorem \ref{thm:introAlgorithm}
    \\
      & Edit & $2^{\tO(\sqrt{\log n})}$ & $n^{1+o(1)}$ 
      & \cite{AO-edit}\\
    \hline
    \multirow{3}{*}{
      \begin{tabular}{@{}l@{}}
        Symmetric \\
        query\\
        complexity
      \end{tabular}
    } 
      & Edit & $n^\eps$ 
      & $\tO(n^{\maxx{1-2\eps,(1-\eps)/2}})$
      & \cite{BEK+03} (fixed $\beta>1$)
    \\
      & Ulam & $O(1)$ & $\tO(\beta+\sqrt{n})$
      & \cite{AN-ulam} 
    \\
      & Ulam+edit & $O(1)$ & \multicolumn{1}{r|}{$\tilde\Omega(\beta+\sqrt{n})$}
      & \cite{AN-ulam} \\
    \hline
    \multirow{5}{*}{
      \begin{tabular}{@{}l@{}}
        Asymmetric \\
        query\\ 
        complexity
      \end{tabular}
    } 
      & Edit & $n^{1/t}$ & $O(\log^{t-1} n)$
      & Theorem \ref{thm:introUpperBound} (fixed $t\in\mathbb N, \beta>1$)
    \\
      & Edit & $n^{1/t}$ & \multicolumn{1}{r|}{$\Omega(\log^{\floor{t}} n)$}
      & Theorem \ref{thm:LB} (fixed $t\notin \mathbb N, \beta>1$) 
    \\
      & Edit & $(\log n)^{1/\eps}$ & $\beta n^{O(\eps)}$
      & Theorem \ref{thm:introUpperBound}
    \\
      & Edit & $(\log n)^{1/\eps}$ & \multicolumn{1}{r|}{$n^{\Omega(\eps/\log\log n)}$}
      & Theorem \ref{thm:LB} (fixed $\beta>1$) 
    \\
      & Ulam & $2+\eps$ & $O_\eps(\beta\log\log\beta\cdot\log n)$
      & \cite{ACCL04} \\
    \hline
  \end{tabular}
  \caption{Known results for
$\DTEP_\beta$ 
and arbitrarily $0<\veps<1$.}
  \label{tbl:results}
\end{table}

\subsection{Connections of Asymmetric Query Model to Other Models}

The asymmetric query model is connected and has implications for two previously studied models, namely the communication complexity model
and the symmetric query model (where the algorithm has query access to both strings).
Specifically, the former is less restrictive than our model
(i.e., easier for algorithms) 
while the latter is more restrictive (i.e., harder for algorithms).
Our upper bound
gives an $O(\beta n^\eps)$ one-way communication complexity protocol for
$\DTEP_\beta$ for polylogarithmic approximation.

\paragraph{Communication Complexity.} In this setting, Alice and Bob
each have a string, and they need to solve the $\DTEP_\beta$ problem
by way of exchanging messages. The measure of complexity is the number
of bits exchanged in order to solve $\DTEP_\beta$ with probability at
least $2/3$.

The best non-trivial upper bound known is $2^{\tilde O(\sqrt{\log
    n})}$ approximation with constant communication via 
\cite{OR-edit,KOR00}. The only known lower bound says that approximation $\alpha$
requires $\Omega(\frac{\log n\ /\ \log\log n}{\alpha})$
communication~\cite{AK07, AJP-sketch}.

The asymmetric model is ``harder'', in the sense that the query
complexity is at least the communication complexity, up to a factor of
$\log|\Sigma|$ in the complexity, since Alice and Bob can simulate the
asymmetric query algorithm. In fact, our upper bound implies a communication
protocol for the same $\DTEP_\beta$ problem with the same complexity,
and it is a one-way communication protocol. Specifically, Alice can
just send the $O(\beta n^\eps)$ characters
queried by the query algorithm in the asymmetric query model.
This is the first communication protocol
achieving polylogarithmic approximation for $\DTEP_\beta$ 
under edit distance with $o(n)$ communication.

\paragraph{Symmetric Query Complexity.} 
In another related model, the measure of complexity 
is the number of characters the algorithm has to query in {\em both} strings
(rather than only in one of the strings).
Naturally, the query complexity in this model is at least
as high as the query complexity in the asymmetric model.
This model has been introduced (for the edit distance)
in~\cite{BEK+03}, and its main advantage is that it leads to {\em
  sublinear-time} algorithms for $\DTEP_\beta$. The algorithm 
of~\cite{BEK+03} makes $\tO(n^{1-2\eps}+n^{(1-\eps)/2})$ queries 
(and runs in the same time), and achieves $n^\eps$ approximation. 
However, it only works for $\beta=O(1)$.

In the symmetric query model, the best query lower bound is of 
$\Omega(\sqrt{n/\alpha})$ for any approximation factor $\alpha>1$ for
both edit and Ulam distance~\cite{BEK+03,AN-ulam}.
The lower bound essentially arises
from the birthday paradox.
Hence, in terms of separating edit distance from the Ulam metric,
this symmetric model can give at most a quadratic separation in the query
complexity (since there exists a trivial algorithm with $2n$
queries). In contrast, in our asymmetric model, 
there is no lower bound based on the birthday paradox, and, in fact, the
Ulam metric admits a constant approximation with $O(\log n)$
queries~\cite{EKKRV00,ACCL04}. Our lower bound for edit distance
is exponentially bigger.

\subsection{Techniques}

This section briefly highlights the main techniques and tools 
used in the course of proving our results.
A more informative proof overview for the algorithmic results,
including the near-linear time algorithm and the query upper bounds, 
appears in Section \ref{sec:ubMain}.
The proof overview for the query lower bounds appears in Section \ref{sec:lb_main}.
The complete proofs are on Sections
\ref{apx:upperBoundFull} and \ref{sec:full_lb}, respectively.

\paragraph{Algorithm and Query Complexity Upper Bound.}

A high-level intuition for the near-linear time algorithm is as
follows. The classical dynamic programming for edit distance runs in time
that is the
product of the lengths of the two strings. It seems plausible that, if we manage to
``compress'' one string to size $n^\eps$, we may be able to compute
the edit distance in time only $n^\eps\cdot n$. Indeed, this is exactly
what we accomplish. Specifically, our ``compression'' is achieved via a sampling
procedure, which subsamples $\approx n^\eps$ positions of $x$, and then
computes $\ed(x,y)$ in time
$n^{1+\eps}$. Of course, the main challenge is, by far,
subsampling $x$ so that the above is possible.

Our asymmetric query upper bound has two major components. 
The first component is 
a {\em characterization} of the edit distance by a different ``distance'',
denoted $\EC$, which approximates $\ed(x,y)$ well. The
characterization is parametrized by an integer parameter $b\ge
2$ governing the following tradeoff: a small $b$ leads to a better approximation,
whereas a large $b$ leads to a faster algorithm.
The second component is a {\em sampling algorithm}
that approximates $\EC$ for some settings of the parameter $b$, up to a
constant  factor, by querying a small number of positions in $x$.

Our characterization is based on a hierarchical decomposition of the edit
distance computation, which is obtained by recursively partitioning 
the string $x$, each time into $b$ blocks.
We shall view this decomposition as a $b$-ary tree. Then, intuitively, the
$\EC$-distance at a node is the sum, over all $b$ children, of the
minima of the $\EC$-distances at these children over a certain range of
displacements (possible ``shifts'' with respect to the other strings). At the leaves (corresponding to single characters
of $x$), the $\EC$-distance is simply the Hamming distance to
corresponding positions in $y$.

We show that our characterization is an $O(\tfrac{b}{\log b}\log n)$
approximation to $\ed(x,y)$. Intuitively, the
characterization manages to break-up the edit distance computation
into {\em independent} distance computations on smaller substrings. The
independence is crucial here as it removes the need to  find a {\em
  global} alignment between the two strings, which is one of the main
reasons why computing edit distance is hard. We note that while the
high-level approach of recursively partitioning 
the strings is somewhat similar to the previous
approaches from~\cite{BEK+03, OR-edit, AO-edit}, the
technical development here is quite different. 
The previous hierarchical approaches all relied on the following
recurrence  relation for the approximation factor $\alpha$:
$$
 \alpha(n)=c\cdot \alpha(n/b)+O(b),
$$
for some $c\ge2$. It is easy to see that one obtains
$\alpha(n)\ge 2^{\Omega(\sqrt{\log n})}$ for any choice of $b\ge 2$. In
contrast, our characterization is much more refined and has {\em
  no multiplicative factor loss}, i.e., $c=1$ and hence  $\alpha(n)=O(b\log_bn)$.
We note that our characterization achieves a {\em logarithmic}
approximation for $b=O(1)$ (although, we do not know efficient
algorithms for this setting of $b$).

The second component of our query algorithm is a careful sampling
procedure that approximates $\EC$-distance up to a constant
factor. The basic idea is to prune the above tree 
by subsampling at each node a subset of its children. In particular, 
for a tree with arity $b=(\log n)^{1/\eps}$, the hope is to subsample
$(\log n)^{O(1)}$ children and use Chernoff-type bounds to argue that
the subsample approximates well the $\EC$-distance at that node. We note
that $\Omega(\log n)$ samples of children seem necessary due to the minimum
operation taken at each node. The estimate at
each node has to hold with high probability so that we can apply the union bound.
After such a pruning of the tree, we would be
left with only $(\log n)^{O(\log_b n)}=n^{O(\eps)}$ leaves, 
i.e., $n^{O(\eps)}$ positions of $x$ to query.

However, this natural approach of subsampling $(\log n)^{O(1)}$ children at
each node does not work when $\beta \gg 1$. Instead, we
develop a {\em non-uniform subsampling technique}: for different nodes
we
subsample children at different, carefully-chosen rates. From a
high-level, our deployed technique is somewhat reminiscent of
the hierarchical decomposition and subsampling technique introduced
by Indyk and Woodruff~\cite{IW05} in the context of sketching and streaming algorithms. 

\paragraph{Query Complexity Lower Bound.}
The gist of our lower bound is designing two ``hard
distributions'' $\D_0$
and $\D_1$, on strings in $\Sigma^n$, for which
it is hard to distinguish with only a few queries to $x$
whether $x\in \D_0$ or $x\in D_1$. 
At the same time, every two strings $x,y$ in the support of the same $\D_i$ 
are at a small edit distance: $\ed(x,y)\le n/(\alpha\beta)$;
but for a mixed pair $x\in \D_0$ and $y\in D_1$, the
distance is large: $\ed(x,y)>n/\beta$.

We start by making the following core observation. Take two random
strings $z_0,z_1\in\zo^n$. Each $\D_i$, $i \in \zo$,
is generated by applying a cyclic shift by a random displacement
$r \in [1,n/100]$ to the corresponding $z_i$. We show that in order to discover,
for an input string, from which $\mathcal D_i$ it came from, 
one has to make at least $\Omega(\log n)$ queries. Intuitively, this
follows from the fact that if the number $q$ of queries is small
($q=o(\log n)$) then
the algorithm's view is close to the uniform distribution on $\zo^q$,
no matter which positions are queried. 
Nevertheless, the edit distance between the two random strings is
likely to be large, 
and a small shift will not change this significantly.

We then amplify the above query lower bound by applying the same idea 
recursively.
In a string generated according to $\mathcal D_i$'s, we replace
every symbol $a \in \zo$  by a random
string selected independently from $\mathcal D_a$.
This way we obtain two distributions on strings of length $n' = n^2$,
that require $\Omega(\log^2 n) = \Omega(\log^2 n')$ queries to be told apart. 
We call the above operation of replacing symbols by strings that come from other distributions a \emph{substitution product}.
Strings created this way consist of $n$ blocks of length $n$ each. 
Intuitively, to distinguish from which of the new distributions an
input string comes from, one has to discover
for at least $\Omega(\log n)$ blocks which distribution $\mathcal D_a$
the respective block comes from.
By applying the recursive step multiple times, we obtain a $2^{\Omega(\frac{\log n}{\log\log n})}$ lower bound for a
\polylg approximation factor. 

To formally prove our result, we develop several tools. First, we need
tools for analyzing the behavior of edit distance under the product
substitution. It turns out that to control edit distance under
the substitution product, we need to work with a large alphabet
$\Sigma$. In the final step of the construction, we map the large
alphabet to sufficiently long random binary strings,
thereby extending the lower bound to the binary alphabet as well.

Second, we need tools for analyzing indistinguishability of our
distributions under a small number of queries. For this, we introduce
a notion of {\em similarity} of distributions. This notion smoothly composes with the
substitution product operation, which amplifies the similarity. We also show
that random acyclic shifts of random strings are likely to produce
strings with high similarity. Finally, we show that if an algorithm is
able to distinguish distributions meeting our similarity notion, then it must make
many queries. We believe that these tools and ideas behind them may
find applications in showing query lower bounds for other problems.

\subsection{Future Directions}

We study a new query model that seems to tap into the hardness 
stemming from ``repetitiveness'' of strings,
obtaining eventually the first algorithm that 
computes a polylogarithmic approximation for edit distance in
near-linear time.
We believe that our techniques may pave the way to significantly improved algorithms
for other tasks involving edit distance, such as the nearest neighbor search.
We mention below a few natural goals for future investigation. 

\myparagraph{Symmetric Model}
Extend our results to the symmetric query model. A lower bound would show a
separation between edit and Ulam distances in this model as well.
It seems plausible that a variation of our hard distribution leads to
a lower bound of the form $n^{1/2+\Omega(1/\log\log n)}$ for
polylogarithmic approximation. The current lower bound is of the form
$\Omega(\sqrt{n/\alpha})$. A query upper bound would likely lead to improved
sub-linear time algorithms. 

\myparagraph{Embedding Lower Bounds}
Is there an $\omega(\log n)$ lower bound for the distortion required to 
embed edit distance into $\ell_1$?
Such a lower bound would answer a well-known open question \cite{Mat-open}.
Note that the core component of our hard distribution,
the shift metric (i.e., hamming cube augmented with cyclic shift operations), 
is known to require distortion $\Omega(\log n)$ \cite{KR06}.

\myparagraph{Communication Complexity}
Prove a communication complexity upper bound of $n^{\eps}$ for all
distance regimes, i.e., independent of $\beta$ (instead of the current
$\beta\cdot n^{\eps}$), for
$\DTEP_\beta$ with polylogarithmic approximation.

\myparagraph{Improved Algorithms}
Tighten the asymmetric query complexity upper bound to
$n^{\tfrac{\eps\log\log\log n}{\log\log n}}$ for approximation $(\log n)^{O(1/\eps)}$, 
perhaps by a more careful subsampling procedure. 
In particular, it seems plausible that one may only sample $(\log\log
n)^{O(1)}$ children at each node,  instead of the present $(\log n)^{O(1)}$.
This may ultimately lead to an algorithm that runs in time $n^{1+o(1)}$ 
and approximates edit distance within a factor of, say, $O(\log^2 n)$.

Perhaps more ambitiously, can one directly use our edit distance characterization 
to compute an $O(\log n)$ approximation in subquadratic time?

\section{Outline of Our Results}

We now sketch the proofs of our results.

\subsection{Outline of the Upper Bound}
\label{sec:ubMain}

In this section, we provide an overview of our algorithmic results, in
particular of the proof of Theorem~\ref{thm:introUpperBound}. Full
statements and proofs of the results appear in
Section~\ref{apx:upperBoundFull}.

Our proof has two major components. The first one is a
characterization of edit distance by a different ``distance'',
denoted $\EC$, which approximates edit distance well.  The second
component is a sampling algorithm that approximates $\EC$ 
up to a constant factor by making a
small number of queries into $x$. We describe each of the components
below. In the following, for a string $x$ and integers $s,t\ge 1$,
$x[s:t]$ denotes the substring of $x$ comprising of $x[s],\ldots,x[t-1]$.

\subsubsection{Edit Distance Characterization: the $\EC$-distance}

Our characterization of $\ed(x,y)$ may be viewed as computation on a tree, where the
nodes correspond to substrings $x[s:s+l]$, for some start position $s\in[n]$ and
length $l\in[n]$. The root is the entire string $x[1:n+1]$. For a node
$x[s:s+l]$, we obtain its children by partitioning $x[s:s+l]$ into $b$
equal-length blocks, $x[s+j\cdot l/b:s+(j+1)\cdot
l/b]$, where $j\in\{0,1,\ldots b-1\}$. Hence $b\ge2$ is the arity of the
tree. The height of the tree is $h \eqdef \log_b n$.
We also use the following notation: for level $i\in\{0,1,\ldots h\}$,
let $l_i \eqdef n/b^i$ be the length of strings at that level. Let $B_i \eqdef \{1,l_i+1,2l_i+1,\ldots\}$
be the set of starting positions of blocks at level $i$. 

The characterization is asymmetric in the two strings and is defined
from a node of the tree to a position $u\in[n]$ of the string $y$.
Specifically, if $i=h$, then the {\em $\EC$-distance} of $x[s]$ to a position
  $u$ is 0 only if $x[s]= y[u]$ and $u\in [n]$, and 1 otherwise.
For $i\in\{0,1,\ldots h-1\}$ and $s\in B_i$, we recursively define the
$\EC$-distance $\EC(i,s,u)$ of $x[s:s+l_i]$ to a position
$u$ as follows. Partition
$x[s:s+l_i]$ into $b$ blocks of length $l_{i+1}=l_i/b$, starting at
positions $s+t_j$, where $t_j \eqdef j\cdot l_{i+1}$, $j\in \{0,1,\ldots b-1\}$. 
Intuitively, we would like to define the $\EC$-distance $\EC(i,s,u)$ 
as the summation of the $\EC$-distances of each block 
$x[s+t_j:s+t_j+l_{i+1}]$ to the corresponding position in $y$, i.e., $u+t_j$.
Additionally, we allow each block to be displaced by some shift $r_j$, 
incurring an additional charge of $|r_j|$ in the $\EC$-distance.  
The shifts $r_j$ are chosen such as to minimize the final distance.
Formally, 
\begin{equation}
\label{eqn:ecDefinition}
\EC(i,s,u) \eqdef \sum_{j=0}^{b-1} \min_{r_j\in \mathbb Z} \EC(i+1,s+t_j,u+t_j+r_j)+\left|r_j\right|.
\end{equation}
The $\EC$-distance from $x$ to $y$ is just
the $\EC$-distance from  $x[1:n+1]$ to position 1, i.e., $\EC(0,1,1)$.  

We illustrate the $\EC$-distance for $b=4$ in
Figure~\ref{fig:ecDistance}. 
Notice that without the shifts (i.e., when all $r_j=0$), the
$\EC$-distance is exactly equal to the Hamming distance between the corresponding strings.
Hence the shifts $r_j$ are what differentiates the Hamming distance and $\EC$-distance.

\begin{figure}[htb]
\begin{center}
\includegraphics[scale=0.45]{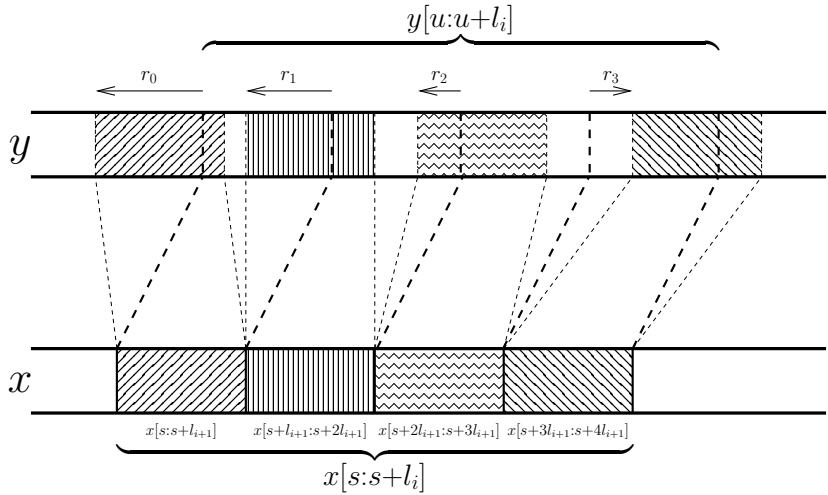}
\end{center}
\caption{Illustration of the $\EC$-distance $\EC(i,s,u)$ for
  $b=4$. The pairs of blocks of the same shading are the blocks whose 
  $\EC$-distance is used for computing $\EC(i,s,u)$.}
\label{fig:ecDistance}
\end{figure}

We prove that the $\EC$-distance is a $O(bh)=O(\tfrac{b}{\log b}\log n)$ approximation to
$\ed(x,y)$ (see Theorem~\ref{thm:ecDistance}). For
$b=2$, the $\EC$-distance is a $O(\log n)$ approximation to
$\ed(x,y)$, but unfortunately, we do not know how to compute it or approximate it well in
better than quadratic time. It is also easy to observe that one can
compute a $1+\eps$
approximation to $\EC$-distance in $\tilde O_\eps(n^2)$ time via a
dynamic programming that considers only $r_j$'s which are
powers of $1+\eps$. 
Instead, we show that, using the query algorithm (described next), we
can compute a $1+\eps$ approximation to 
$\EC$-distance for $b=(\log n)^{O(1/\eps)}$ in $n^{1+\eps}$ time.

\subsubsection{Sampling Algorithm}

We now describe the ideas behind our sampling algorithm. The sampling
algorithm approximates the $\EC$-distance between $x$ and $y$ up to a
constant factor. The query complexity is $Q\le \beta\cdot (\log
n)^{O(h)}=\beta\cdot (\log n)^{\log_b n}$ for distinguishing $\EC(0,1,1)> n/\beta$
from $\EC(0,1,1)\le n/(2\beta)$. For the rest of this overview, it
is instructive to think about the setting where $\beta=n^{0.1}$ and
$b=n^{0.01}$, although our main result actually follows by setting
$b=(\log n)^{O(1/\eps)}$.

The idea of the algorithm is to prune the characterization tree, and in particular prune the
children of each node. If we retain only
$\polylog n$ children for each node, we would obtain the claimed $Q\le (\log
n)^{O(h)}$ leaves at the bottom, which correspond to the sampled
positions in $x$. The main challenge is how to perform this pruning.

A natural idea is to uniformly subsample $\polylog n$ out of $b$ children at each node, and use Chernoff-type concentration bounds to argue that
Equation~\eqref{eqn:ecDefinition} may be approximated only from the
$\EC$-distance estimates of the subsampled children. Note that, 
since we use the minimum operator at each node, we have to 
aim, at each node, for an estimate that holds with high probability.

How much do we have to subsample at each node? The ``rule of thumb''
for a Chernoff-type bound to work well is as follows. Suppose we have
quantities $a_1,\ldots a_m\in [0,\rho]$ respecting an upper bound
$\rho>0$, and let $\sigma=\sum_{j\in [m]} a_j$. 
Suppose we subsample several $j\in[m]$ to form
a set $J$. Then, in order to estimate $\sigma$ well (up
to a small multiplicative factor) from $a_j$ for $j\in J$, we need to subsample essentially
a total of $|J|\approx \tfrac{\rho}{\sigma}\cdot m\log m$ positions $j\in [m]$. We call
this Uniform Sampling Lemma (see Lemma~\ref{lem:uniformSampling}
for complete statement).

With the above ``sampling rule'' in mind, we can readily see that, at the top
of the tree, until a level $i$, where $l_i=n/\beta$, there is no
pruning that may be done (with the notation from above, we have
$\rho=l_i=n/\beta$ and $\sigma=n/\beta$).
However, we hope to prune the tree  at the subsequent levels.

It turns out that such pruning is not possible as
described. Specifically, consider a node $v$ at
level $i$ and its children $v_j$, for $j=0,\ldots b-1$. Suppose each child contributes a
distance $a_j$ to the sum $\EC$ at node $v$ (in
Equation~\eqref{eqn:ecDefinition}, for fixed $u$). Then, because of the bound on length of
the strings, we have that $a_j\le l_{i+1}=(n/\beta)/b$. At the same
time, for an average node $v$, we have $\sum_{j=0}^{b-1} a_j\approx
l_i/\beta=n/\beta^2$. By the Uniform Sampling Lemma from above, we need to take a
subsample of size
$|J|\approx \tfrac{n/(\beta b)}{n/\beta^2}\cdot b\log b=\beta\log b$. If
$\beta$ were constant, we would obtain $|J|\ll b$ and hence prune the
tree (and, indeed, this approach works for $\beta\ll b$). However,
once $\beta\gg b$, such pruning does not seem possible. In fact, one
can give counter-examples where such pruning approach fails to
approximate the $\EC$-distance.

To address the above challenge, we develop a way to 
prune the tree {\em non-uniformly}. Specifically, for different nodes we will
subsample its children at different, well-controlled rates. In fact, for each node we
will assign a ``precision'' $w$ with the requirement that a node $v$,
at level $i$, with precision $w$, must estimate its $\EC$-distances to
positions $u$ up to an {\em additive error} $l_i/w$. The pruning and
assignment of precision will proceed top-bottom, starting with
assigning a precision $4\beta$ to the root node. Intuitively, the higher
the precision of a node $v$, the denser is the subsampling in the subtree
rooted at $v$.

Technically, our main tool is a {\em Non-uniform Sampling
  Lemma}, which we use to assign the necessary precisions to nodes. It
may be stated as follows (see 
Lemma~\ref{lem:nonUniformSampling} for a more complete
statement).  The lemma says that there exists some distribution $\W$ and a
reconstruction algorithm $R$ such that the following two conditions
hold:
\begin{itemize}
\item
Fix some $a_j\in[0,1]$
for $j\in[m]$, with $\sigma=\sum_j a_j$. Also, pick $w_j$ i.i.d.~from
the distribution $\W$ for 
each $j\in[m]$. Let $\hat a_j$ be estimators of $a_j$, up
to an additive error of $1/w_j$, i.e., $|a_j-\hat a_j|\le
1/w_j$. Then the algorithm $R$, given $\hat a_j$ and $w_j$ for $j\in[m]$,
outputs a value that is inside $[\sigma-1,\sigma+1]$, with high probability.
\item
$\EE{w\in \W}{w}=\polylog m$.
\end{itemize}
To internalize this statement, fix 
$\sigma=10$, and consider two extreme cases. At one extreme,
consider some set of 10 $j$'s such that $a_j=1$, and all the others
are 0. In this case, the previous uniform subsampling rule does
not yield any savings (to continue the parallel, uniform sampling can be seen as
having $w_j=m$ for the sampled $j$'s and $w_j=1$ for the
non-sampled $j$'s). Instead, it would suffice to take all
$j$'s, but approximate them up to ``weak'' (cheap) precision (i.e.,
set $w_j\approx 100$ for all $j$'s). At
the other extreme is the case when $a_j=10/m$ for all $j$. In this case,
subsampling would work but then one requires a much ``stronger''
(expensive) precision, of the order of $w_j\approx m$. These 
examples show that one cannot choose all $w_j$ to be equal. If $w_j$'s
are too small, it is impossible to estimate $\sigma$. If $w_j$'s are too
big, the expectation of $w$ cannot be bounded by $\polylog m$, and the subsampling is too expensive.

The above lemma is somewhat inspired by the sketching and streaming
technique introduced by Indyk and Woodruff~\cite{IW05} (and used for the $F_k$ moment
estimation), where one partitions elements 
$a_j$ by weight level, and then performs corresponding subsampling in
each level. Although related, our approach to the above lemma differs:
for example, we avoid any definition of the weight level (which was usually
the source of some additional complexity of the use of the
technique). For completeness, we mention that the distribution $\W$ is
essentially the distribution with probability distribution function
$f(x)=\nu/x^2$ for $x\in[1,m^3]$ and a normalization constant
$\nu$. The algorithm $R$ essentially uses the samples that were
(in retrospect) well-approximated, i.e., $\hat a_j\gg 1/w_j$, in order
to approximate $\sigma$.

In our $\EC$-distance estimation algorithm, we use both uniform and non-uniform
subsampling lemmas at each node to both prune the tree and assign the
precisions to the subsampled children. We note that
the lemmas may be used to 
obtain a multiplicative $(1+\eps')$-approximation for arbitrary small $\eps'>0$
for each node.
To obtain this, it is necessary to
use $\eps \approx \eps'/\log n$, since over $h\approx \log n$ levels, we collect a multiplicative approximation factor of $(1+\eps)^h$, which remains constant only as
long as $\eps=O(1/h)$.

\subsection{Outline of the Lower Bound} 
\label{sec:lb_main}

In this section we outline the proof of Theorem \ref{thm:LB}.
The full proof appears in Section~\ref{sec:full_lb}.
Here, we focus on the main ideas, skipping or simplifying 
some of the technical issues.

As usual, the lower bound is based on constructing ``hard distributions'', 
i.e., distributions (over inputs) that cannot be distinguished using few queries,
but are very different in terms of edit distance.
We sketch the construction of these distributions in Section \ref{sec:sketchHardDist}. The full construction appears in Section \ref{sect:hard_distro}.
In Section \ref{sec:sketchDistinguish}, we sketch the machinery that
we developed to prove that distinguishing these distributions requires 
many queries; the details appear in Section~\ref{sec:tools_distinguish}.
We then sketch in Section \ref{sec:sketchED} the tools needed to prove
that the distributions are indeed very different in terms of edit distance;
the detailed version appears in Section \ref{sec:tools_edit_distance}.

\subsubsection{The Hard Distributions} \label{sec:sketchHardDist}

We shall construct two distributions $\mathcal D_0$ and $\mathcal D_1$ 
over strings of a given length $n$. The distributions satisfy the following properties. 
First, every two strings in the support of the same distribution $\mathcal D_i$, denoted $\supp(\D_i)$, are close in edit distance.
Second, every string in $\supp(\D_0)$ is far in edit distance from every string in $\supp(\D_1)$. 
Third, if an algorithm correctly distinguishes (with probability at least $2/3$) whether its input string is drawn from $\mathcal D_0$ or from $\mathcal D_1$, it must make many queries to the input.

Given two such distributions, we let $x$ be any string from $\supp(\mathcal D_0)$. This string is fully known to the algorithm. The other string $y$, to which the algorithm only has query access, is drawn from either $\mathcal D_0$ or $\mathcal D_1$. Since distinguishing the distributions apart requires many queries to the string, so does approximating edit distance between $x$ and $y$.

\paragraph{Randomly Shifted Random Strings.}
The starting point for constructing these distributions is the following idea. Choose at random two base strings $z_0,z_1\in \zo^n$. These strings are likely to satisfy some ``typical properties'', e.g.\ be far apart in edit distance (at least $n/10$). Now let each $\mathcal D_i$ be the distribution generated by selecting a cyclic shift of $z_i$ by $r$ positions to the right, where $r$ is a uniformly random integer between $1$ and $n/1000$. Every two strings in the same $\supp(\D_i)$ are at distance at most $n/500$, because a cyclic shift by $r$ positions can be produced by $r$ insertions and $r$ deletions. At the same time, by the triangle inequality,
every string in $\supp(D_0)$ and every string in $\supp(\D_1)$ must be at distance at least $n/10 - 2 \cdot n/500 \ge n/20$.

How many queries are necessary to learn whether an input string is drawn from $\mathcal D_0$ or from $\mathcal D_1$? If the number $q$ of queries is small, then the algorithm's view
is close to a uniform distribution on $\zo^q$
under both $\D_0$ and $\D_1$.
Thus, the algorithm is unlikely to distinguish the two distributions with probability significantly higher than $1/2$. This is the case because each base string $z_i$ is chosen at random and because we consider many cyclic shifts of it. Intuitively, even if the algorithm knows $z_0$ and $z_1$, the random shift makes the algorithm's view a nearly-random pattern, because of the random design of $z_0$ and $z_1$. Below we introduce rigorous tools for such an analysis. They prove, for instance, that even an adaptive algorithm for this case, and in particular every algorithm that distinguishes edit distance $\le n/500$ and $\ge n/20$, must make $\Omega(\log n)$ queries.

One could ask whether the $\Omega(\log n)$ lower bound for the number of queries in this construction can be improved.
The answer is negative, because for a sufficiently large constant $C$, by querying any consecutive $C\log n$ symbols of $z_1$, one obtains a pattern that most likely does not occur in $z_0$, and therefore, can be used to distinguish between the distributions.
This means that we need a different construction to show a superlogarithmic lower bound.

\paragraph{Substitution Product.} 
We now introduce the \emph{substitution product}, which plays an important role in our lower bound construction.
Let $\mathcal D$ be a distribution on strings in $\Sigma^m$. For each $a \in \Sigma$, let $\mathcal E_a$ be a distribution on $(\Sigma')^{m'}$, and denote their entire collection by $\mathcal E \eqdef (\mathcal E_a)_{a \in \Sigma}$. Then the substitution product $\mathcal D \sprod \mathcal E$ is the distribution generated
by drawing a string $z$ from $\mathcal D$, and independently replacing every symbol $z_i$ in $z$ by a string $B_i$ drawn from $\mathcal E_{z_i}$.

Strings generated by the substitution product consist of $m$ blocks. Each block is independently drawn from one of the $\mathcal E_a$'s, and a string drawn from $\mathcal D$ decides which $\mathcal E_a$ each block is drawn from.

\paragraph{Recursive Construction.}
We build on the previous construction with two random strings shifted at random, and extend it by introducing recursion. For simplicity, we show how this idea works for two levels of recursion. We select two random strings $z_0$ and $z_1$ in $\zo^{\sqrt{n}}$. We use a sufficiently small positive constant $c$ to construct two distributions $\mathcal E_0$ and $\mathcal E_1$. $\mathcal E_0$ and $\mathcal E_1$ are generated by taking a cyclic shift of $z_0$ and $z_1$, respectively, by $r$ symbols to the right, where $r$ is a random integer between 1 and $c\sqrt{n}$. 
Let $\mathcal E \eqdef (\mathcal E_i)_{i \in \zo}$.

Our two hard distributions on $\zo^n$ are $\mathcal D_0 \eqdef \mathcal E_0 \sprod \mathcal E$, and
$\mathcal D_1 \eqdef \mathcal E_1 \sprod \mathcal E$. 
As before, one can show that distinguishing a string drawn from $\mathcal
E_0$ and a string drawn from $\mathcal E_1$ is likely to require
$\Omega(\log n)$ queries. In other words, the algorithm has to
\emph{know} $\Omega(\log n)$ symbols from a string selected from one
of $\mathcal E_0$ and $\mathcal E_1$. Given the recursive structure of
$\mathcal D_0$ and $\mathcal D_1$, the hope is that distinguishing them
requires at least $\Omega(\log^2 n)$ queries, because at least
intuitively, the algorithm ``must'' know for at least $\Omega(\log n)$
blocks which $\E_i$ they come from, each of the blocks requiring
$\Omega(\log n)$ queries. Below, we describe techniques that we use to
formally prove such a lower bound. It is straightforward to show that every two strings drawn from the same $\mathcal D_i$ are at most $4cn$ apart. It is slightly harder to prove that strings drawn from $\mathcal D_0$ and $\mathcal D_1$ are far apart.
The important ramification is that for some constants $c_1$ and $c_2$, distinguishing edit distance $<c_1n$ and $>c_2n$ requires $\Omega(\log^2 n)$ queries, where one can make $c_1$ much smaller than $c_2$.
For comparison, under the Ulam metric, $O(\log n)$ queries suffice for 
such a task (deciding whether distance between a known string and 
an input string is $<c_1n$ or $>c_2n$, assuming $2c_1 < c_2$ \cite{ACCL04}).

To prove even stronger lower bounds, we apply the substitution product
several times, not just once. Pushing our approach to the limit, we prove that distinguishing edit distance $O(n/\polylog n)$ from $\Omega(n)$ requires $n^{\Omega\left({1}/{\log \log n}\right)}$ queries. In this case, $\Theta\left({\log n}/{\log \log n}\right)$ levels of recursion are used.
One slight technical complication arises in this case. Namely, we need to work with 
a larger alphabet (rather than binary). Our result holds true for the binary
alphabet nonetheless, since we show that one can effectively reduce the larger alphabet to the binary alphabet.

\subsubsection{Bounding the Number of Queries}  \label{sec:sketchDistinguish}

We start with definitions. Let $\mathcal D_0$, \ldots, $\mathcal D_k$ be distributions on the same finite set $\Omega$ with $p_1,\ldots,p_k:\Omega \to [0,1]$ as the corresponding probability mass functions.
We say that the distributions are \emph{$\alpha$-similar}, where $\alpha \ge 0$, if for every $\omega \in \Omega$,
$$(1-\alpha)\cdot\max_{i=1,\ldots,k} p_i(\omega) \le \min_{i=1,\ldots,k} p_i(\omega).$$

For a distribution $\mathcal D$ on $\Sigma^n$ and $Q \subseteq[n]$, we write $\mathcal D|_Q$ to denote the
distribution created by projecting every element of $\Sigma^n$ to its coordinates in $Q$.
Let this time $\mathcal D_1$, \ldots, $\mathcal D_k$ be probability distributions on $\Sigma^n$.
We say that they are \emph{uniformly $\alpha$-similar} if for every subset $Q$ of $[n]$,
the distributions $\mathcal D_1|_Q$, \ldots, $\mathcal D_k|_Q$ are $\alpha|Q|$-similar.
Intuitively, think of $Q$ as a sequence of queries that the algorithm makes. If the distributions
are uniformly $\alpha$-similar for a very small $\alpha$, and $|Q| \ll 1/\alpha$,
then from the limited point of view of the algorithm (even an adaptive one), 
the difference between the distributions is very small.

In order to use the notion of uniform similarity for our construction, we prove the following three key lemmas.

\myparagraph{Uniform Similarity Implies a Lower Bound on the Number of Queries (Lemma~\ref{lemma:SimilarityAlgorithms})}
This lemma formalizes the ramifications of uniform $\alpha$-similarity for a pair of distributions.
It shows that if an algorithm (even an adaptive one) distinguishes the two distributions with probability at least $2/3$, then it has to make at least $1/(6\alpha)$ queries. The lemma implies that it suffices to bound the uniform similarity in order to prove a lower bound on the number of queries.

The proof is based on the fact that for every setting of the algorithm's random bits, the algorithm can be described as a decision tree of depth $q$, if it always makes at most $q$ queries. Then, for every leaf, the probability of reaching it does not differ by more than a factor in $[1-\alpha q,1]$ between the two distributions. This is enough to bound the probability the algorithm outputs the correct answer for both the distributions.

\myparagraph{Random Cyclic Shifts of Random Strings Imply Uniform Similarity (Lemma~\ref{lem:RotationSimilarity})} 
This lemma constructs block-distributions that are uniformly similar
using cyclic shifts of random base strings.
It shows that if one takes $n$ random base strings in $\Sigma^n$ and creates $n$ distributions by shifting each of the strings by a random number of indices in $[1,s]$, then with probability at least $2/3$ (over the choice of the base strings) the created distributions are uniformly $O(1/\log_{|\Sigma|}\frac{s}{\log n})$-similar.

It is easy to prove this lemma for any set $Q$ of size 1. In this case, every shift gives an independent random bit, and the bound directly follows from the Chernoff bound. A slight obstacle is posed by the fact that for $|Q|\ge 2$, sequences of $|Q|$ symbols produced by different shifts are not necessarily independent, since they can share some of the symbols. To address this issue, we show that there is a partition of shifts into at most $|Q|^2$ large sets such that no two shifts of $Q$ in the same set overlap. Then we can apply the Chernoff bound independently to each of the sets to prove the bound.

In particular, using this and the previous lemmas, one can show the result claimed earlier that shifts of two random strings in $\zo^n$ by an offset in $[1,cn]$ produce distributions that require $\Omega(\log n)$ queries to be distinguished. It follows from the lemma that the distributions are likely to be uniformly $O(1/\log n)$-similar.

\myparagraph{Substitution Product Amplifies Uniform Similarity (Lemma~\ref{lemma:SimilarityMultiplication})} 
Perhaps the most surprising property of uniform similarity is that it nicely composes with the substitution product. Let $\mathcal D_1$, \ldots, $\mathcal D_k$ be uniformly $\alpha$-similar distributions on $\Sigma^n$.
Let $\mathcal E = (\mathcal E_a)_{a \in \Sigma}$, where $\E_a$, $a \in \Sigma$, are uniformly $\beta$-similar distributions on $(\Sigma')^{n'}$. The lemma states that $\mathcal D_1 \sprod \mathcal E$, \ldots, $\mathcal D_k \sprod \mathcal E$ are uniformly $\alpha \beta$-similar.

The main idea behind the proof of the lemma is the following.
Querying $q$ locations in a string that comes from $\mathcal D_i \sprod \mathcal E$, we can see a difference between distributions in at most $\beta q$ blocks in expectation. Seeing the difference is necessary to discover which $\mathcal E_j$ each of the blocks comes from. Then only these blocks can reveal the identity of $\mathcal D_i \sprod \mathcal E$, and the difference in the distribution if $q'$ blocks are revealed is bounded by $\alpha q'$. 

The lemma can be used to prove the earlier claim that the two-level construction produces distributions that require $\Omega(\log^2 n)$ queries to be told apart.

\subsubsection{Preserving Edit Distance} \label{sec:sketchED}

It now remains to describe our
tools for analyzing the edit distance between strings generated by our
distributions. All of these tools are collected in Section~\ref{sec:tools_edit_distance}. In most cases we focus in our analysis on $\edd$, which is the version of edit distance that only allows for insertions and deletions. It clearly holds that $\ed(x,y) \le \edd(x,y) \le 2\cdot\ed(x,y)$, and this connection is tight enough for our purposes. An additional advantage of $\edd$ is that for any strings $x$ and $y$, $2\lcs(x,y) + \edd(x,y) = |x| + |y|$.

We start by reproducing a well known bound on the longest common substring of randomly selected strings (Lemma~\ref{lem:RandomStrings}). It gives a lower bound on $\lcs(x,y)$ for two randomly chosen 
strings. The lower bound then implies that the distance between two strings chosen at random is large, especially for a large alphabet.

Theorem~\ref{thm:sprodAdditive} shows how the edit distance between two strings in $\Sigma^n$ changes when we substitute every symbol with a longer string using a function $B:\Sigma \to (\Sigma')^{n'}$.
The relative edit distance (that is, edit distance divided by the length of the strings) shrinks by an additive term that polynomially depends on the maximum relative length of the longest common string between $B(a)$ and $B(b)$ for different $a$ and $b$. It is worth to highlight the following two issues:
\begin{itemize}
 \item We do not need a special version of this theorem for distributions. It suffices to first bound edit distance for the recursive construction when instead of strings shifted at random, we use strings themselves. Then it suffices to bound by how much the strings can change as a result of shifts (at all levels of the recursion) to obtain desired bounds.
 \item The relative distance shrinks relatively fast as a result of substitutions. This implies that we have to use an alphabet of size polynomial in the number of recursion levels. The alphabet never has to be larger than polylogarithmic, because the number of recursion levels is always $o(\log n)$.
\end{itemize}

Finally, Theorem~\ref{thm:sprodFar} and Lemma \ref{lem:randomFar2} 
effectively reduce the alphabet size,
because they show that a lower bound for the binary alphabet 
follows immediately from the one for a large alphabet, 
with only a constant factor loss in the edit distance. It turns out that it suffices to map every element of the large alphabet $\Sigma$ to a random string of length $\Theta(\log |\Sigma|)$ over the binary alphabet.

The main idea behind proofs of the above is that strings constructed 
using a substitution product are composed of rather rigid blocks,
in the sense that every alignment between two such strings,
say $x\sprod\E$ and $y\sprod\E$, must respect (to a large extent) 
the block structure, in which case one can extract from it 
an alignment between the two initial strings $x$ and $y$.

\section{Fast Algorithms via Asymmetric Query Complexity}
\label{apx:upperBoundFull}

In this section we describe our near-linear time algorithm for
estimating the edit distance between two strings. As we mentioned in the introduction, the
algorithm is obtained from an efficient query algorithm.

The main result of this section is the following query complexity
upper bound theorem, which is a full version of
Theorem~\ref{thm:introUpperBound}. It implies our near-linear time
algorithm for polylogarithmic approximation (Theorem~\ref{thm:introAlgorithm}).
\begin{theorem}
\label{thm:upperBound}
Let $n\ge 2$, $\beta=\beta(n)\ge 2$, and integer $b=b(n)\ge2$
be such that $(\log_b n)\in\N$.

There is an algorithm solving
$\DTEP_{\beta}$ with approximation $\alpha=O(b\log_b n)$  and 
$\beta\cdot (\log n)^{O(\log_b n)}$ queries into $x$. The algorithm runs in
$n\cdot (\log n)^{O(\log_b n)}$ time.

For every constant $\beta=O(1)$ and integer $t\ge 2$,
there is an algorithm for solving $\DTEP_\beta$ with $O(n^{1/t})$
approximation and $O(\log n)^{t-1}$ queries. The
algorithm runs in $\tilde O(n)$ time.
\end{theorem}
In particular, note that we obtain Theorem~\ref{thm:introAlgorithm} by
setting $b=(\log n)^{c/\eps}$ for a suitably high constant $c>1$.

The proof is partitioned in three stages. (The first stage corresponds
to the first ``major component'' mentioned in Introduction, and
Section~\ref{sec:ubMain}, and the next two stages correspond to 
the second ``major component''.)
In the first stage, we describe a characterization of edit
distance by a different quantity, namely
$\EC$-distance, which approximates edit distance well. The
characterization is parametrized by an integer parameter $b\ge
2$. A small $b$ leads to a small approximation factor (in fact, as
small as $O(\log n)$ for $b=2$),
whereas a large $b$ leads to a faster algorithm.
In the second stage, we show how one can design a sampling algorithm
that approximates $\EC$-distance for some setting of the parameter $b$, up to a
constant factor, by making a small
number of queries into $x$. In the third stage, we show how to use the query
algorithm to obtain a near-linear time algorithm for edit distance
approximation.

The three stages are described in the following three sections, and
all together give the proof of Theorem~\ref{thm:upperBound}.

\subsection{Edit Distance Characterization: the $\EC$-distance}

Our characterization  may be viewed as computation on a tree, where the
nodes correspond to substrings $x[s:s+l]$, for some start position $s\in[n]$ and
length\footnote{We remind that the notation $x[s:s+l]$ corresponds to characters
  $x[s],x[s+1],\ldots x[s+l-1]$. More generally, $[s:s+l]$ stands for
  the interval $\{s,s+1,\ldots , s+l-1\}$. This convention simplifies
 subsequent formulas.} $l\in[n]$. The root is the entire string $x[1:n+1]$. For a node
$x[s:s+l]$, the children are blocks $x[s+j\cdot l/b:s+(j+1)\cdot
l/b]$, where $j\in\{0,1,\ldots b-1\}$, and $b$ is the arity of the
tree. The $\EC$-distance for the node $x[s:s+l]$ is defined
recursively as a function of the distances of its children.
Note that the characterization is asymmetric in the two strings.

Before giving the definition we establish further notation. We fix the
arity $b\ge2$ of the tree, and let $h\eqdef\log_b n\in \N$ be the height of the tree.
Fix some tree level $i$ for $0\le i\le h$. Consider some substring $x[s:s+l_i]$
at level $i$, where $l_i \eqdef n/b^i$. Let
$B_i\eqdef\{1,l_i+1,2l_i+1,\ldots\}$ be the set of starting positions of
blocks at level $i$.

\begin{definition}[$\EC$-distance]
\label{def:AdistRecursive}
Consider two strings $x,y$ of length $n\ge2$.
Fix $i\in \{0,1,\ldots h\}$, $s\in B_i$, and a position $u\in \mathbb Z$.

If $i=h$, then the {\em $\EC$-distance} of $x[s:s+l_i]$ to the position
  $u$ is 1 if $u\not\in [n]$ or $x[s]\neq y[u]$, and 0 otherwise.

For $i\in\{0,1,\ldots h-1\}$, we recursively define the {\em $\EC$-distance} $\EC_{x,y}(i,s,u)$ of $x[s:s+l_i]$ to the position $u$ as follows. Partition
$x[s:s+l_i]$ into $b$ blocks of length $l_{i+1}=l_i/b$, starting at
positions $s+jl_{i+1}$, where $j\in \{0,1,\ldots b-1\}$. 
Then
$$
\EC_{x,y}(i,s,u) \eqdef \sum_{j=0}^{b-1} \min_{r_j\in \mathbb Z} \EC_{x,y}(i+1,s+jl_{i+1},u+jl_{i+1}+r_j)+\left|r_j\right|.
$$

The $\EC$-distance from $x$ to $y$ is just the $\EC$-distance from
$x[1:n+1]$ to position 1, i.e., $\EC_{x,y}(0,1,1)$. 
\end{definition}

We illustate the $\EC$-distance for $b=4$ in
Figure~\ref{fig:ecDistance}. Since $x$ and $y$ will be clear from the context, we will just 
use the notation $\EC(i,s,u)$ without indices $x$ and $y$.

The main property of the $\EC$-distance is that it gives a good approximation
to the edit distance between $x$ and $y$, as quantified in the
following theorem, which we prove below.

\begin{theorem}[Characterization]
\label{thm:ecDistance}
For evry $b\ge2$ and two strings $x,y\in \Sigma^n$, the $\EC$-distance between $x$ and $y$ is a
$6\cdot\tfrac{b}{\log b}\cdot\log n$ approximation to the edit
distance between $x$ and $y$. 
\end{theorem}

We also give an alternative, equivalent
definition of the $\EC$-distance between $x$ and $y$. It is motivated
by considering the matching (alignment) induced by the $\EC$-distance
when computing $\EC(0,1,1)$. In particular, when computing
$\EC(0,1,1)$ recursively, we can consider all the ``matching
positions'' (positions $u+jl_{i+1}+r_j$ for $r_j$'s achieving the minimum).
We denote by $Z$ a vector of integers $z_{i,s}$, indexed by $i\in\{0,1,\ldots h\}$ and
$s\in B_i$, where $z_{0,1}=1$ by convention. The coordinate $z_{i,s}$
should be understood as the
position to which we match the substring $x[s:s+l_i]$ in the
calculation of $\EC(0,1,1)$. Then we define the cost
of $Z$ as
$$
\cost(Z) \eqdef \sum_{i=0}^{h-1} \sum_{s\in B_i} \sum_{j=0}^{b-1} |z_{i,s}+jl_{i+1}-z_{i+1,s+jl_{i+1}}|.
$$

The cost of $Z$ can be seen as the sum of the displacements $|r_j|$
that appear in the calculation of the $\EC$-distance from
Definition~\ref{def:AdistRecursive}. The following claim asserts an
alternative definition of the $\EC$-distance.

\begin{claim}[Alternative definition of $\EC$-distance]
\label{clm:ecDist2Z}
The $\EC$-distance between $x$ and $y$ is the minimum of
\begin{equation}
\label{eqn:ecDistViaZ}
\cost(Z)+\sum_{s\in [n]} \Ham(x[s],y[z_{h,s}])
\end{equation}
over all choices of the vector $Z=(z_{i,s})_{i\in\{0,1,\ldots h\}, s\in
  B_i}$ with $z_{0,1}=1$, where $\Ham(\cdot,\cdot)$ is the Hamming distance,
namely $H(x[s], y[z_{h,s}])$ is 1 if $z_{h,s}\not\in [n]$ or $x[s]\neq
y[z_{h,s}]$, and 0 otherwise.
\end{claim}

\begin{proof}
The quantity~\eqref{eqn:ecDistViaZ} simply unravels the recursive
formula from Definition~\ref{def:AdistRecursive}.
The equivalence between them follows from the fact that
$|z_{i,s}+jl_{i+1}-z_{i+1,s+jl_{i+1}}|$ directly corresponds to
quantities $|r_j|$
in the $\EC_{x,y}(i,s,z_{i,s})$ definition, which appear in the computation on the tree,
and the $\sum_{s\in [n]} \Ham(x[s],y[z_{h,s}])$ term corresponds to
the summation of $\EC_{x,y}(h,s,z_{h,s})$ over all $s \in [n]$.
\end{proof}

We are now ready to prove Theorem~\ref{thm:ecDistance}.

\begin{proof}[Proof of Theorem~\ref{thm:ecDistance}]
Fix $n,b\ge2$ and let $h \eqdef \log_b n$. 
We break the proof into two parts, an upper bound and a lower bound 
on the $\EC$-distance (in terms of edit distance).
They are captured by the following two lemmas, which we shall prove shortly.

\begin{lemma}
\label{lem:edLEecdist}
The $\EC$-distance between $x$ and $y$ is
at most $3hb\cdot\ed(x,y)$.
\end{lemma}

\begin{lemma}
\label{lem:Zdist2ed}
The edit distance $\ed(x,y)$ is at most twice the
$\EC$-distance between $x$ and $y$.
\end{lemma}

Combining these two lemmas gives
$\tfrac{1}{2}\ed(x,y)\le \EC_{x,y}(0,1,1)\le 5hb\cdot\ed(x,y)$,
which proves Theorem~\ref{thm:ecDistance}.
\end{proof}

We proceed to prove these two lemmas.

\begin{proof}[Proof of Lemma~\ref{lem:edLEecdist}]
Let $A:[n]\to [n]\cup \{\bot\}$ be an optimal alignment from $x$ to
$y$. Namely $A$ is such that:
\begin{itemize}
 \item If $A(s)\neq \bot$, then $x[s]=y[A(s)]$.
 \item If $A(s_1)\neq \bot$, $A(s_2)\neq \bot$, and $s_1 < s_2$, then $A(s_1) < A(s_2)$.
 \item $L \eqdef |A^{-1}(\bot)|$ is minimized.
\end{itemize}
Note that $n-L$ is the length of the Longest Common
Subsequence (LCS) of $x$ and $y$. It clearly holds that
$\tfrac{1}{2}\ed(x,y)\le L\le \ed(x,y)$.

To show an upper bound on the $\EC$-distance, we use the alternative
characterization from Claim~\ref{clm:ecDist2Z}. Specifically, we show
how to construct a vector $Z$ proving that the $\EC$-distance is small.

At each level $i\in \{1,2,\ldots h\}$, for each block $x[s:s+l_i]$ where
$s\in B_i$, we
set $z_{i,s} \eqdef A(j)$, where $j$ is the smallest
integer $j\in[s:s+l_i]$ such that $A(j)\neq \bot$ (i.e., to match a block
we use the first in it that is aligned under the alignment
$A$). If no such $j$ exists, then $z_{i,s} \eqdef z_{i-1,s'}+(s-s')$,
where $s' \eqdef l_{i-1}\cdot \lfloor (s-1)/l_{i-1}\rfloor+1$,
that is, $s'$ is such that $x[s':s'+l_{i-1}]$ is the parent of $x[s:s+l_{i}]$ in the tree.

Note that it follows from the definition of $z_{h,s}$ and $L$ that $\sum_{s\in[n]} \Ham(x[s],
y[z_{h,s}])=L$. It remains to bound the other term 
$\cost(Z)$ in the alternative definition of $\EC$-distance.

To accomplish this, for every $i \in \{0,1,2,\ldots,h-1\}$ and $s \in B_i$, we define $d_{i,s}$ as the maximum of 
$|z_{i,s}+jl_{i+1}-z_{i+1,s+jl_{i+1}}|$ over $j \in \{0,\ldots, b-1\}$. Although
we cannot bound each $d_{i,s}$ separately, we bound the sum of
$d_{i,s}$ for each level $i$.

\begin{claim}
\label{clm:sumDis}
For each $i\in\{0,1,\ldots h\}$, we have that $\sum_{s\in B_i} d_{i,s}\le 2L$. 
\end{claim}
\begin{proof}
We shall prove that each $d_{i,s}$ is bounded by $X_{i,s} + Y_{i,s}$, where
$X_{i,s}$ and $Y_{i,s}$ are essentially the number
of unmatched positions in $x$ and in $y$, respectively, that contribute to
$d_{i,s}$. We then argue that both $\sum_{s\in B_i} X_{i,s}$ and $\sum_{s\in B_i}
Y_{i,s}$ are bounded by $L$, thus completing the proof of the claim.

Formally, let $X_{i,s}$ be the number of positions $j\in
[s:s+l_i]$ such that $A(j)=\bot$. If $X_{i,s}=l_i$, then clearly
$d_{i,s}=0$. 
It is also easily verified that if $X_{i,s}=l_i-1$, then $d_{i,s}\le l_i-1$. 
In both cases, $d_{i,s}\le X_{i,s}$, and we also set $Y_{i,s} \eqdef 0$.

If $X_{i,s}\le l_i-2$, let $j'$ be
the largest integer $j'\in[s:s+l_i]$ for which $A(j')\neq \bot$ (note
that $j'$ exists and it is different from the smallest such possible integer, 
which was called $j$ when we defined $z_{i,s}$, because $X_{i,s}\le l_i-2$).
In this case, let $Y_{i,s}$ be $A(j')-z_{i,s}+1-(l_i-X_{i,s})$,
which is the number of positions in $y$
between $z_{i,s}$ and $A(j')$ (inclusive) that are not aligned under $A$.
Let $\Delta_{i,s,j} \eqdef z_{i,s}+jl_{i+1}-z_{i+1,s+jl_{i+1}}$ for $j \in \{0,\ldots,b-1\}$.
By definition, it holds $d_{i,s} = \max_j |\Delta_{i,s,j}|$.
Now fix $j$.
If $\Delta_{i,s,j} \ne 0$, then there is
an index $k \in [s+jl_{i+1}:s+(j+1)l_{i+1}]$ such that $A(k) = z_{i+1,s+jl_{i+1}}$.
If $\Delta_{i,s,j} > 0$ (which corresponds to a shift to the left), then at least $\Delta_{i,s,j}$ indices $j'\in[s:k]$
are such that $A(j') = \bot$, and therefore, $|\Delta_{i,s,j}| \le X_{i,s}$.
If $\Delta_{i,s,j} < 0$ (which corresponds to a shift to the right), then at least $|\Delta_{i,s,j}|$ positions
in $y$ between $z_{i,s}$ and $z_{i+1,s+jl_{i+1}}$ are not aligned in $A$.
Thus, $|\Delta_{i,s,j}| \le Y_{i,s}$.

In conclusion, for every $s \in B_i$, $d_{i,s} \le X_{i,s}+Y_{i,s}$.
Observe that $\sum_{s\in B_i} X_{i,s} =L$ and $\sum_{s\in B_i} Y_{i,s} \le L$ 
(because they correspond to distinct positions in $x$ and in $y$ that are
not aligned by $A$). 
Hence, we obtain that $\sum_{s\in
  B_i} d_{i,s}\le \sum_{s\in B_i} X_{i,s} + Y_{i,s}\le 2L$.
\end{proof}

We now claim that $\cost(Z)\le 2hbL$. Indeed, consider a
block $x[s:s+l_i]$ for some $i\in\{0,1,\ldots h-1\}$ and $s\in
B_i$, and one of its children $x[s+jl_{i+1}:s+(j+1)l_{i+1}]$ for $j\in
\{0,1,\ldots b-1\}$. The contribution of this child to the sum $\cost(Z)$ is
$
|z_{i,s}+jl_{i+1}-z_{i+1,s+jl_{i+1}}|\le d_{i,s}$ by definition.
Hence, using Claim~\ref{clm:sumDis}, we conclude that
$$
\cost(Z)
\le 
\sum_{i=0}^{h-1}\sum_{s\in B_i} \sum_{j=0}^{b-1} d_{i,s}
\le 
\sum_{i=0}^{h-1}\sum_{s\in B_i} d_{i,s}\cdot b
\le
h\cdot 2L\cdot b.
$$

Finally, by Claim~\ref{clm:ecDist2Z}, we have that the $\EC$-distance
between $x$ and $y$ is at most $2hbL+L\le 2hb\cdot
\ed(x,y)+\ed(x,y)\le 3hb\cdot \ed(x,y)$.
\end{proof}

\begin{proof}[Proof of Lemma~\ref{lem:Zdist2ed}]
We again use the alternative characterization given by 
Claim~\ref{clm:ecDist2Z}. Let $Z$ be the vector obtaining the minimum
of Equation~\eqref{eqn:ecDistViaZ}. Define, for $i\in\{0,1,\ldots
h\}$ and $s\in B_i$,
$$
\delta_{i,s}\eqdef
\sum_{s'\in [s:s+l_i]}\Ham(x[s'],y[z_{h,s'}])
+
\sum_{i':i\le i'< h} \ 
\sum_{s'\in B_{i'}\cap [s:s+l_i]}\ 
\sum_{j=0}^{b-1} \left|z_{i',s'}+jl_{i'+1}-z_{i'+1,s'+jl_{i'+1}}\right|
.
$$
Note that $\delta_{0,1}$ equals the $\EC$-distance by Claim~\ref{clm:ecDist2Z}. Also, we
have the following inductive equality for $i\in\{0,1,\ldots h-1\}$ and
$s\in B_i$:
\begin{equation}
\label{eqn:recursiveDelta}
\delta_{i,s}=\sum_{j=0}^{b-1} \left(\delta_{i+1,s+jl_{i+1}}+|z_{i,s}+jl_{i+1}-z_{i+1,s+jl_{i+1}}|\right).
\end{equation}

We now prove inductively for
$i\in \{0,1,2\ldots h\}$ that for each $s\in B_i$, the length of the LCS
of $x[s:s+l_i]$ and $y[z_{i,s}:z_{i,s}+l_i]$ is at least $l_i-\delta_{i,s}$.

For the base case, when $i=h$, the inductive hypothesis is
trivially true. If $x[s]=y[z_{i,s}]$, then the LCS is of length 1
and $\delta_{h,s}=0$. If $x[s]\neq y[z_{i,s}]$, then
the LCS is of length 0 and $\delta_{h,s}=1$.

Now we prove the inductive hypothesis for $i\in \{0,1,\ldots h-1\}$,
assuming it holds for $i+1$.  Fix a string $x[s:s+l_i]$, and let $s_j=s+jl_{i+1}$ for
$j\in\{0,1,\ldots b-1\}$. By the inductive hypothesis, for each
$j\in\{0,1,\ldots b-1\}$, the length of the LCS between
$x[s_j:s_j+l_{i+1}]$ and $y[z_{i+1,s_j}:z_{i+1,s_j}+l_{i+1}]$ is at least $l_{i+1}-\delta_{i+1,s_j}$. 
In this case, the substring in $y$ starting at $z_{i,s}+jl_{i+1}$, namely
$y[z_{i,s}+jl_{i+1}:z_{i,s}+(j+1)l_{i+1}]$, has an LCS with
$x[s_j:s_j+l_{i+1}]$ of length at least
$l_{i+1}-\delta_{i+1,s_j}-|z_{i,s}+jl_{i+1}-z_{i+1,s_j}|$. Thus, by Equation~\eqref{eqn:recursiveDelta},
the LCS of $x[s:s+l_i]$ and $y[z_{i,s}:z_{i,s}+l_i]$ is of length at
least
$$
\sum_{j=0}^{b-1} \left( l_{i+1}-\delta_{i+1,s_j}-|z_{i,s}+jl_{i+1}-z_{i+1,s_j}| \right)
=l_i-\delta_{i,s},
$$
which finishes the proof of the inductive step.

For $i=0$, this implies that $\ed(x,y)\le
2\delta_{0,1}=2\EC_{x,y}(0,1,1)$.
\end{proof}

\subsection{Sampling Algorithm}

We now describe the sampling and estimation algorithms that are used
to obtain our query complexity upper bounds. In particular, our
algorithm approximates the $\EC$-distance defined in the previous
section. The guarantee of our algorithms is that the output $\hat\EC$
satisfies $(1-o(1))\EC(0,1,1)-n/\beta \le \hat\EC\le (1+o(1))\EC(0,1,1)+n/\beta$.
This is clearly sufficient to distinguish between $\EC(0,1,1)\le
n/\beta$ and $\EC(0,1,1)\ge 4n/\beta$. After presenting the algorithm,
we prove its correctness and prove that it only samples $\beta\cdot
n^{O(\eps)}$ positions of $x$ in order to make the decision.

\subsubsection{Algorithm Description}

We now present our sampling algorithm, as well as the estimation
algorithm, which given $y$ and the sample of $x$, decides $\DTEP_\beta$.

\paragraph{Sampling algorithm.}
To subsample $x$, we start by partitioning $x$ recursively into blocks as defined
in Definition~\ref{def:AdistRecursive}. In particular, we fix a tree
of arity $b$, indexed by pairs $(i,s)$ for $i\in\{0,1,\ldots h\}$, and
$s\in B_i$.
At each level $i=0,\ldots h$, we have a subsampled set $C_i\subseteq
B_i$ of vertices at that level of
the tree. The set $C_i$ is obtained from the previous one by
extending $C_{i-1}$ (considering all the children), and a careful
subsampling procedure. In fact, for each element
in $C_i$, we also assign a number $w\ge 1$, representing a
``precision'' and
describing how well we want to estimate the $\EC$ distance at
that node, and hence governing the subsampling of the subtree rooted
at the node.

Our sampling algorithm works as follows. We use a (continuous)
distribution $\W$ on $[1,n^3]$, which we define later, in
Lemma~\ref{lem:nonUniformSampling}. 
\begin{algorithm}[!h]
\caption{Sampling Algorithm}
\label{alg:sampling}
Take $C_0$ to be the root vertex (indexed $(i,s)=(0,1)$),
with precision $w_{(0,1)}=\beta$.
\\
\For{each level $i=1,\ldots,h$, we construct $C_i$ as follows}{
Start with $C_i$ being empty. 
\\
\For{each node $v=(i-1,s)\in C_{i-1}$}{
Let $w_v$ be its precision, and set $p_v=\tfrac{w_v}{b}\cdot O(\log^3 n)$. 
\\
If $p_v\ge 1$, then set $J_v=\{(i,s+jl_i)\mid 0\le j<b\}$ to be the
set of all the $b$ children of $v$, and add them to $C_{i}$, each with precision $p_v$.
\\
Otherwise, when $p_v<1$, sample each of the $b$ children of $v$
with probability $p_v$, to form a set $J_v\subseteq \{i\}\times ([s:s+l_{i-1}]\cap B_{i})$.
For each $v'\in J_v$, draw $w_{v'}$ i.i.d. from $\cal W$, and add node
$v'$ to $C_i$ with precision $w_{v'}$.
\\
}
}
Query the characters $x[s]$ for all $(h,s)\in C_h$ --- this
is the output of the algorithm.
\end{algorithm}

\paragraph{Estimation Algorithm.} We compute a value $\tau(v,z)$, for
each node $v\in \cup_i C_i$ and position $z\in[n]$, such that $\tau(v,z)$
is a good approximation ($1+o(1)$ factor) to the $\EC$-distance of the node
$v$ to position $z$.

We also use a ``reconstruction algorithm'' $R$, defined
in Lemma~\ref{lem:nonUniformSampling}. It takes as input at most $b$
quantities, their precision, and outputs a positive number.
\begin{algorithm}
\caption{Estimation Algorithm}
\label{alg:estimation}
For each sampled leaf $v=(h,s)\in C_h$ and $z\in[n]$ we set $\tau(v,z)=\Ham(x[s],y[z])$.
\\
\For{each level $i=h-1,j-2,\ldots,0$, position $z\in[n]$, and node $v\in C_i$ with
  precision $w_v$} {
We apply the following procedure $P(v,z)$ to obtain
$\tau(v,z)$.
\\
For each $v'\in J_v$, where $v'=(i+1,s+jl_{i+1})$ for
some $0\le j<b$, let 
$$
\delta_{v'}\eqdef \min_{k:|k|\le n} \tau(v', z+jl_{i+1}+k)+|k|.
$$
\\
If $p_v\ge 1$, then let $\tau(v,z)=\sum_{v'\in J_v} \delta_{v'}$.
\\
If $p_v<1$, set $\tau(v,z)$ to be the output of the algorithm $R$ on the vector
$(\tfrac{\delta_{v'}}{l_{i+1}})_{v'\in J_v}$ with precisions
$(w_{v'})_{v'\in J_v}$, multiplied by $l_{i+1}/p_v$.
\\
}
The output of the algorithm is $\tau(r,1)$ where $r=(0,1)$ is the root of the tree.
\end{algorithm}

\subsubsection{Analysis Preliminaries: Approximators and a Concentration Bound}

We use the following approximation notion that captures both an
additive and a multiplicative error. 
For convenience, we work with factors $e^\eps$ instead of usual
$1+\eps$.

\begin{definition}
\label{def:approximator}
Fix $\rho>0$ and some $f\in[1,2]$.
For a quantity $\tau\ge 0$, we call its {\em $(\rho,
  f)$--approximator} any quantity $\hat \tau$ such that
$\tau/f-\rho\le \hat \tau\le f\tau+\rho$.
\end{definition}

It is immediate to note the following {\em additive property}: if
$\hat\tau_1, \hat\tau_2$ are $(\rho,f)$-approximators to
$\tau_1,\tau_2$ respectively, then $\hat\tau_1+\hat\tau_2$ is a
$(2\rho,f)$-approximator for $\tau_1+\tau_2$. Also, there's a
composion property: if $\hat \tau'$ is an $(\rho',f')$-approximator to
$\hat\tau$, which itself is a $(\rho,f)$-approximator to $\tau$, then
$\hat\tau'$ is a $(\rho'+f'\rho,ff')$-aproximator to $\tau$.

The definition is motivated by the following concentration statement
on the sum of random variables. The statement is an
immediate application of the standard Chernoff/Hoeffding bounds.

\begin{lemma}[Sum of random variables]
\label{lem:sumConcentration}
Fix $n\in \N$, $\rho>0$, and error probability $\delta$. Let
$Z_i\in [0,\rho]$ be independent random variables,
and let $\zeta>0$ be a sufficiently large absolute constant.
Then for every $\eps\in(0,1)$,
the summation $\sum_{i\in [n]} Z_i$ is a $(\zeta\rho\tfrac{\log 1/\delta}{\eps^2},
e^\eps)$-approximator to $\EE{}{\sum_{i\in[n]} Z_i}$, with probability $\ge1-\delta$.
\end{lemma}

\begin{proof}[Proof of Lemma~\ref{lem:sumConcentration}]
By rescaling, it is sufficient to prove the claim for $\rho=1$.
Let $\mu=\EE{}{\sum_{i\in [n]} Z_i}$. If $\mu>\tfrac{\zeta}{4}\cdot\tfrac{\log
 1/\delta}{\eps^2}$, then, a standard application of the Chernoff implies
that $\sum_i Z_i$ is a $e^\eps$ approximation to $\mu$, with
$\ge1-\delta$ probability, for some sufficiently high $\zeta>0$.

Now assume that $\mu\le\tfrac{\zeta}{4}\cdot\tfrac{\log 1/\delta}{\eps^2}$.
We use the following variant of the Hoeffding inequality, which can be
derived from~\cite{Hoeffding63}.
\begin{lemma}[Hoeffding bound]\label{lem:hoeffding}
Let $Z_i$ be $n$ independent random variables such that $Z_i\in[0,1]$,
and $\EE{}{\sum_{i\in[n]} Z_i}=\mu$. Then, for any $t>0$, we have
that
$
\Pr\left[\sum_i Z_i\ge t\right]\le e^{-(t-2\mu)}.
$
\end{lemma}
We apply the above lemma for $t=\zeta\cdot\tfrac{\log
  1/\delta}{\eps^2}$. We obtain that $\Pr[\sum_i Z_i\ge t]\le
e^{-t/2}=e^{-\Omega(\log 1/\delta)}<\delta$, which completes the proof that $\sum_i
Z_i$ is a $(\zeta\tfrac{\log 1/\delta}{\eps^2}, e^\eps)$-approximator to
$\mu$ (when $\rho=1$).
\end{proof}

\subsubsection{Main Analysis Tools: Uniform and Non-uniform Sampling Lemmas}

We present our two main subsampling lemmas that are applied, recursively, at each node of 
the tree.  The first lemma, on Uniform Sampling, is a simple
Chernoff bound in a suitable regime. 

The second lemma, called Non-uniform Sampling Lemma, is the heart of
our sampling, and is inspired by a sketching/streaming technique introduced in~\cite{IW05}
for optimal estimation of $F_k$ moments in a stream. Although a
relative of their method, our lemma is different both in
intended context and actual technique.
We shall use the constant $\zeta>0$ coming from 
Lemma \ref{lem:sumConcentration}.

\begin{lemma}[Uniform Sampling]
\label{lem:uniformSampling}
Fix $b\in \N$, $\eps>0$, and error probability $\delta>0$.
Consider some $a_j$, $j\in[b]$, such that
$a_j\in[0,1/b]$. For arbitrary $w\in [1,\infty)$, construct the set $J\subseteq [b]$ by
subsampling each $j\in[b]$ with probability $p_w=\min\{1,\tfrac{w}{b}\cdot
\zeta\tfrac{\log 1/\delta}{\eps^2}\}$.
Then, with probability at least $1-\delta$, the value
$\tfrac{1}{p_w}\sum_{j\in J} a_j$ is a
$(1/w, e^\eps)$-approximator to $\sum_{j\in [b]} a_j$, and
$|J|\le O(w\cdot \tfrac{\log 1/\delta}{\eps^2})$.
\end{lemma}
\begin{proof}
If $p_w=1$, then $J=[b]$ and there is nothing to prove; so assume that
$p_w=\tfrac{w}{b}\cdot \zeta\tfrac{\log 1/\delta}{\eps^2}<1$ for the
rest.  

The bound on $|J|$ follows from a standard application of the Chernoff
bound: 
$\EE{}{|J|}=p_wb\le O(w\cdot\tfrac{\log 1/\delta}{\eps^2})$,
hence the probability that $|J|$ exceeds twice the quantity
is at most $e^{-\Omega(\log 1/\delta)}\leq \delta/2$.

We are going to apply Lemma~\ref{lem:sumConcentration} to the
variables $Z_j=a_j/p_w\cdot \chi[j\in J]$, where the indicator variable $\chi[j\in J]$ is 1
iff $j\in J$. Note that $0\le Z_j\le \tfrac{\eps^2}{w\cdot \zeta\log 1/\delta}$. We thus obtain that $\sum_{j\in [b]} Z_j$ is
a $(\tfrac{\zeta\eps^{-2} \log 1/\delta}{w\cdot \zeta\eps^{-2}\log 1/\delta}, e^\eps)$-approximator, and hence $(1/w, e^\eps)$-approximator, to $\EE{}{\sum_j
  Z_j}=\sum_{j\in [b]} p_w\cdot \tfrac{a_j}{p_w}=\sum_{j\in [b]} a_j$.
\end{proof}

We now present and prove the Non-uniform Sampling Lemma.

\begin{lemma}[Non-uniform Sampling]
\label{lem:nonUniformSampling}
Fix integers $n\le N$, approximation $\eps>0$, factor $1<f<1.1$, error
probability $\delta>0$, and an
``additive error bound'' $\rho>6n/\eps/N^3$. There exists a
distribution $\cal W$ on the real interval $[1,N^3]$ with $\EE{w\in
  \W}{w}\le O(\tfrac{1}{\rho}\cdot \tfrac{\log 1/\delta}{\eps^3}\cdot
\log N)$, as well as a ``reconstruction algorithm'' $R$,  with the following property.

Take arbitrary $a_i\in[0,1]$, for $i\in[n]$, and let $\sigma=\sum_{i\in[n]}
a_i$. Suppose one draws $w_i$ i.i.d. from $\cal W$, for each
$i\in[n]$, and let $\hat a_i$ be a
$(1/w_i,f)$-approximator of $a_i$.
Then, given $\hat a_i$ and $w_i$ for all
$i\in[n]$, the algorithm $R$  generates a
$(\rho,f\cdot e^{\eps})$-approximator to
$\sigma$, with probability at least $1-\delta$. 
\end{lemma}

For concreteness, we mention that $\W$ is the maximum of
$O(\tfrac{1}{\rho}\cdot\tfrac{\log
  1/\delta}{\eps^3})$ copies of the (truncated) distribution
$1/x^2$ (essentially equivalent to a distribution of $x$ where
the logarithm of $x$ is distributed geometrically).

\begin{proof}
We start by describing the distribution $\W$ and the algorithm
$R$. Fix $k=\tfrac{2\zeta}{\rho}\cdot\tfrac{\log
  1/\delta}{(\eps/2)^3}$. We first describe a related 
distribution: let $\W_1$ be distribution on $x$
such that the pdf function is $p_1(x)=\nu/x^2$ for $1\le x\le
N^3$ and $p_1(x)=0$ otherwise, where $\nu=(\int_{1}^\infty
p_1(x)\dx)^{-1}=(1-1/N^3)^{-1}$ is a normalization constant. Then $\W$
is the distribution of $x$ where we choose $k$ i.i.d. variables $x_1,\ldots
x_k$ from $\W_1$ and then set $x=\max_{i\in[k]} x_i$. Note that the
pdf of $\W$ is $p(x)=\nu^k \tfrac{k}{x^2}(1-1/x)^{k-1}$.

The algorithm $R$ works as follows. For each $i\in[n]$, we define $k$ 
``indicators''
$s_{i,j}\in\{0,1/k\}$ for $j\in[k]$. Specifically, we generate the set of
random variables $w_{i,j}\in \W_1$, $j\in[k]$, conditioned on the fact
that $\max_{j\in[k]}w_{i,j}=w_i$.
Then, for each $i\in[n],j\in[k]$, we set $s_{i,j}=1/k$
if $\hat a_i\ge t/w_i$ for $t=3/\eps$, and $s_{i,j}=0$
otherwise. Finally, we set
$s=\sum_{i\in [n], j\in[k]} s_{i,j}$ 
and the algorithm outputs $\hat\sigma=st/\nu$
(as an estimate for $\sigma$). 

We note that the variables $w_{i,j}$ could be thought as being chosen
i.i.d. from $\W_1$. For each, the value $\hat a_i$ is an
$(1/w_{i,j},f)$-approximator to $a_i$ since $\hat a_i$ is a
$(1/\max_j w_{i,j},f)$-approximator to $a_i$. 

It is now easy to bound $\EE{w\in\W}{w}$.
Indeed, we have $\EE{w\in \W_1}{w}= \int_{1}^{N^3}
x\cdot \nu/x^2\dx\le O(\log N)$. Hence
$\EE{w\in \W}{w}\le \sum_{j\in [k]} \EE{w\in
  \W_1}{w}\le O(k\log N)=O(\tfrac{1}{\rho}\cdot\tfrac{\log
  1/\delta}{\eps^3}\cdot \log N)$.

We now need to prove that $\hat\sigma$ is an approximator to $\sigma$, with
probability at least $1-\delta$. We first compute the
expectation of $s_{i,j}$, for each $i\in [n], j\in[k]$. This
expectation depends on the approximator values $\hat a_i$,
which itself may depend on $w_i$. Hence we can only give upper and lower
bounds on the expectation $\EE{}{s_{i,j}}$. Later, we want to apply a
  concentration bound on
the sum of $s_{i,j}$. Since $s_{i,j}$ may be interdependent, we
will apply the concentration bound on the upper/lower bounds of
$s_{i,j}$ to give bounds on $s=\sum s_{i,j}$.

Formally, we define random variables $\overline{s}_{i,j},
\underline{s}_{i,j}\in\{0,1/k\}$. We set $\overline{s}_{i,j}=1/k$ iff $w_{i,j}\ge
(t-1)/(fa_i)$, and 0 otherwise. Similarly, we set
$\underline{s}_{i,j}=1/k$ iff $w_{i,j}<f(t+1)/a_i$, and 0 otherwise.
We now claim that 
\begin{equation}
\label{eqn:sUBLB}
\underline{s}_{i,j}\le s_{i,j}\le \overline{s}_{i,j}.
\end{equation}
Indeed, if $s_{i,j}=1/k$, then $\hat a_{i}\ge t/w_{i,j}$, and hence,
using the fact that $\hat a_i$ is a $(1/w_{i,j}, 
f)$-approximator to $a_i$, we have $w_{i,j}\ge (t-1)/(fa_i)$, or
$\overline{s}_{i,j}=1/k$. Similarly, if $s_{i,j}=0$, then $\hat a_{i}<
t/w_{i,j}$, and hence $w_{i,j}<f(t+1)/a_i$, or
$\underline{s}_{i,j}=0$. Note that each collection $\{\overline{s}_{i,j}\}$
 and $\{\underline{s}_{i,j}\}$ is a collection of independent random variables.

We now bound $\EE{}{\overline{s}_{i,j}}$ and
$\EE{}{\underline{s}_{i,j}}$. For the first quantity, we have:
$$
\EE{}{\overline{s}_{i,j}}= \int_{(t-1)/(fa_i)}^{N^3} \tfrac{1}{k}p_1(x)\dx\le
\tfrac{fa_i}{k(t-1)}\int_1^{\infty}\nu/x^2\dx=\nu/k\cdot\tfrac{fa_i}{t-1}.
$$
For the second quantity, we have:
$$
\EE{}{\underline{s}_{i,j}}= \int_{f(t+1)/a_i}^{N^3} p_1(x)\dx=\nu/k\cdot(\tfrac{a_i/f}{t+1}-1/N^3).
$$

Finally, using Eqn.~\eqref{eqn:sUBLB} and the fact that
$\EE{}{s}=\sum_{i,j} \EE{}{s_{i,j}}$, we can bound 
$\EE{}{\hat\sigma}=\EE{}{st/\nu}$ as follows:
$$
\tfrac{t}{f(t+1)}\sum_{i\in[n]} a_i-nt/N^3
\le
\tfrac{t}{\nu}\sum_{i,j} \EE{}{\underline{s}_{i,j}}
\le
\EE{}{ts/\nu}
\le
\tfrac{t}{\nu}\sum_{i,j} \EE{}{\overline{s}_{i,j}}
\le
 f\sum_{i\in[n]}a_i\cdot \tfrac{t}{t-1}.
$$
Since each $\overline{s}_{i,j},\underline{s}_{i,j}\in[0,1/k]$ for
$k=O(\tfrac{t}{\rho}\cdot\tfrac{\log 1/\delta}{\eps^2})$, 
we can apply Lemma~\ref{lem:sumConcentration} to obtain a high
concentration bound. For the upper bound, we obtain, with probability
at least $1-\delta/2$:
$$
ts/\nu
\le 
e^{\eps/2}\cdot \EE{}{t/\nu\cdot \sum_{i,j} s_{i,j}}+\rho
\le
e^{\eps/2}\cdot f\sum a_i\cdot \tfrac{t}{t-1}+\rho
\le e^\eps\cdot f\cdot\sigma+\rho.
$$
Similarly, for the lower bound, we obtain, with probability at least
$1-\delta/2$:
$$
ts/\nu\ge e^{-\eps/2}\cdot (\sum a_i\cdot \tfrac{t}{f(t+1)}-nt/N^3)-\rho/2
\ge e^{-\eps}/f\cdot \sigma-\rho,
$$
using that $\rho/2\ge nt/N^3$.
This completes the proof that $\hat\sigma$ is a
$(\rho,f\cdot e^\eps)$-approximator to $\sigma$, with probability at least $1-\delta$.
\end{proof}

\subsubsection{Correctness and Sample Bound for the Main Algorithm}

Now, we prove the correctness of the algorithms
\ref{alg:sampling},~\ref{alg:estimation} and bound its
query complexity. We note that we use Lemmas~\ref{lem:uniformSampling}
and~\ref{lem:nonUniformSampling} with $\delta=1/n^3$, $\eps=1/\log n$,
and $N=n$ (which in particular, completely determine the distribution $\W$
and algorithm $R$ used in the algorithms~\ref{alg:sampling} and~\ref{alg:estimation}).

\begin{lemma}[Correctness]
\label{lem:correctness}
For $b=\omega(1)$, the output of the Algorithm~\ref{alg:estimation} (Estimation), is a
$(n/\beta, 1+o(1))$-approximator to 
the $\EC$-distance from $x$ to $y$, w.h.p.
\end{lemma}

\begin{proof}
From a high level view, we prove inductively from $i=0$ to $i=h$ that
expanding/subsampling the current $C_i$ gives a
good approximator, namely a $e^{O((h-i)/\log n)}$ factor
approximation, with probability at least
$1-i/n^{\Omega(1)}$. 
Specifically, at each step of the
induction,  we expand and subsample each node
from the current $C_i$ to form the set $C_{i+1}$ and use
Lemmas~\ref{lem:uniformSampling} 
and~\ref{lem:nonUniformSampling} to show that we don't loose on the
approximation factor by more than $e^{O(1/\log n)}$.

In order to state our main inductive hypothesis, we define a
hybrid distance, where the $\EC$-distance of nodes at high levels (big
$i$) is
computed standardly (via Definition~\ref{def:AdistRecursive}), and the
$\EC$-distance of the low-level nodes is estimated via sets $C_i$.
Specifically, for fixed $f\in[1,1.1]$, and $i\in \{0,1,\ldots h\}$, we
define the following {\em $(C_0,C_1\ldots
C_i,f)$-$\EC$-distance}. For each vertex $v=(i,s)$ such that $v\in
C_i$ has precision $w_v$, and $z\in[n]$,
let $\tau_i(v,z)$ to be some $(l_{i}/w_v,f)$--approximator to the
distance $\EC(i,s,z)$. Then, iteratively for $i'=i-1,i-2,\ldots,0$, 
for all $v\in
C_{i'}$ and $z\in[d]$, we compute $\tau_i(v,z)$ by applying the
procedure $P(v,z)$ (defined in the Algorithm~\ref{alg:estimation}),
using $\tau_i$ instead of $\tau$.

We prove the following inductive hypothesis, for some
suitable constants $t=2$ and $r=\Theta(1)$ (sufficiently high $r$ suffices).
\begin{itemize}
\item[$\IH_i$:]
For any $f\in[1,1.1]$, the $(C_0,C_1,\ldots C_i,f)$--$\EC$--distance is a
$(n/\beta, f\cdot e^{i\cdot t/\log
  n})$-approximator to the $\EC$--distance from $x$ to $y$, with
probability at least $1-i\cdot e^{-r\log n}$.
\end{itemize}

Base case is $i=0$, namely that $(C_0,f)$-$\EC$-distance is a
$(n/\beta,f)$-approximator to the $\EC$-distance between $x$ and $y$.
This case follows immediately from the
definition of the $(C_0,f)$-$\EC$-distance and the initialization step
of the Sampling Algorithm.

Now we prove the inductive hypothesis $\IH_{i+1}$, assuming $\IH_i$ holds for
some given $i\in\{0,1,\ldots h-1\}$. We remind that we defined the quantity
$\tau_{i+1}(v,z)$, for all $v\in
C_{i+1}\subseteq \overline C_i$, where $\overline
C_i=\{(i+1,s+jl_{i+1})\mid (i,s)\in C_i,j\in\{0,\ldots b-1\}\}$ and
$z\in [n]$, to be a $(l_{i+1}/w_v,f)$--approximator of the corresponding
$\EC$-distance, namely $\EC(v,z)$. The plan is to prove that, for all
$v\in C_i$ with precision $w_v$, the quantity $\tau_{i+1}(v,z)$ is a
$(l_{i}/w_v,f\cdot e^{2/\log n})$--approximator to
$\EC(v,z)$ with good probability --- which we do in the claim below. Then, by definition of
$\tau_i$ and $\IH_i$, this implies that $\tau_{i+1}((0,1),1)$ is equal
to the $(C_0,\ldots C_i,f\cdot e^{2/\log n}\cdot e^{i\cdot t/\log
  n})$--$\EC$--distance, and hence is
a $(n/\beta,f\cdot e^{(2+it)/\log n})$--approximator to the
$\EC$--distance from $x$ to $y$. This will complete the proof of $\IH_{i+1}$.
We now prove the main technical step of the above plan.

\begin{claim}
\label{clm:stepP}
Fix $v\in C_i$ with precision $w\eqdef w_v$, where $v=(i,s)$, and some $z\in[n]$. For $j\in
\{0,\ldots b-1\}$, let $v_j$ be the $j^{th}$ child of
$v$; i.e., $v_j=(i+1,s+jl_{i+1})$. For
$v_j\in C_{i+1}$ with precision $w_j\eqdef w_{v_j}$, and $z'\in[n]$, let
$\tau_{i+1}(v_j, z')$ be  
a $(l_{i+1}/w_j,f)$--approximator to $\EC(v_j,z')$.

Apply procedure $P(v,z)$ using $\tau_{i+1}(v_j,z')$ estimates, and let $\delta$ be the
output. Then $\delta$ is a
$(l_i/w, fe^{2/\log n})$--approximator to $\EC(v,z)$, with probability at least
$1-e^{-\Omega(\log n)}$.
\end{claim}

\begin{proof}
For each $v_j\in J_v$, where $J_v$ is as defined in
Algorithm~\ref{alg:sampling}, we define the following quantities: 
$$
\delta_{v_j}\eqdef\min_{k:|k|\le n} \EC(v_j,z+jl_{i+1}+k)+|k|
\qquad
\hat\delta_{v_j}\eqdef\min_{k:|k|\le n} \tau_{i+1}(v_j,z+jl_{i+1}+k)+|k|.
$$

It is immediate to see that $\hat\delta_{v_j}$ is a
$(l_{i+1}/w_j,f)$--approximator to $\delta_{v_j}$ by the definition of
$\tau_{i+1}$.

If $p_v\ge 1$, then we have that $w_j=\tfrac{w}{b}\cdot O(\log^3 n)$ for all
$v_j\in J_v$. Then, by the additive property of $(l_{i+1}/w_j,
f)$--approximators, $\delta=\sum_{v_j\in J_v}\hat\delta_{v_j}$ 
is a
$(l_i/w,f)$--approximator to $\sum_{v_j\in
  J_v}\delta_{v_j}=\EC(v,z)$. 

Now suppose $p_v<1$. Then, by Lemma~\ref{lem:uniformSampling},
$\delta'=\tfrac{1}{p_v}\sum_{v_j\in J_v} \delta_{v_j}$ is a
$(l_i/2w,e^{1/\log n})$--approximator to $\sum_{j=0}^{b-1} \delta_{v_j}=\EC(v,z)$, with high
probability. Furthermore, by Lemma~\ref{lem:nonUniformSampling} for
$\rho=1$, 
since $w_j\in \W$ are i.i.d. and $\tfrac{\hat\delta_{v_j}}{l_{i+1}}$ are
each an
$(1/w_j,f)$--approximator to $\tfrac{\delta_{v_j}}{l_{i+1}}$ 
respectively, then $R$ outputs a value
$\delta''$ that is a
$(1,f\cdot e^{1/\log n})$--approximator to 
$\sum_{v_j\in J_v} \tfrac{\delta_{v_j}}{l_{i+1}}=\tfrac{p_v}{l_{i+1}}\delta'$.
In other words,
$\delta=\tfrac{l_{i+1}}{p_v}\delta''$ is a $(l_{i+1}/p_v,f\cdot
e^{1/\log n})$--approximator to
$\delta'$. Since $l_{i+1}/p_v\le l_i/(3w)$, combining the two
approximator guarantees, we obtain that
$\delta$ is a $(l_i/w,f\cdot e^{2/\log n})$--approximator to
$\EC(v,z)$, w.h.p.
\end{proof}

We now apply a union bound over all $v\in C_i$ and $z\in[n]$, and use
the above Claim~\ref{clm:stepP}.
We now apply $\IH_i$ to deduce that $\tau_{i+1}((0,1),1)$ is a
$(n/\beta,f\cdot e^{ti/\log n}\cdot e^{2/\log n})$--approximator with
probability at least
$$
1-ie^{-r\log n}-e^{-\Omega(\log n)}\ge 1-(i+1)e^{-r\log n},
$$
for some suitable $r=\Theta(1)$. This proves $\IH_{i+1}$.

Finally we note that $\IH_{h}$ implies that $(C_0,\ldots
C_h,f)$--$\EC$--distance is a $(n/\beta, f\cdot e^{th/\log n})$--approximator
to the $\EC$--distance between $x$ and $y$. We conclude the lemma with
the observation that our Estimation Algorithm~\ref{alg:estimation}
outputs precisely the $(C_0,\ldots C_h,1)$--$\EC$--distance.
\end{proof}

It remains to bound the number of positions that 
Algorithm~\ref{alg:estimation} queries into $x$.

\begin{lemma}[Sample size]
\label{lem:sampleSize}
The Sampling Algorithm queries $Q_b=\beta (\log n)^{O(\log_b n)}$ positions
of $x$, with probability at least $1-o(1)$.
When $b=n^{1/t}$ for fixed constant $t\in \N$ and $\beta=O(1)$, we
have $Q_b=(\log n)^{t-1}$ with probability at least 2/3.
\end{lemma}
\begin{proof}
We prove by induction, from $i=0$ to $i=h$, that $\EE{}{|C_i|}\le \beta\cdot
(\log n)^{ic}$, and $\EE{}{\sum_{v\in C_i} w_v}\le \beta\cdot
(\log n)^{ic+5}$ for a suitable $c=\Theta(1)$. The base case of $i=0$ is immediate by the
initialization of the Sampling Algorithm~\ref{alg:sampling}. Now we prove the inductive
step for $i$, assuming the inductive hypothesis for $i-1$. 
By Lemma~\ref{lem:uniformSampling}, $\EE{}{|C_{i}|}\le
\EE{}{\sum_{v\in C_{i-1}} w_v}\cdot O(\log^3 n)\le \beta (\log n)^{ic}$ by
the inductive hypothesis. Also, by Lemma~\ref{lem:nonUniformSampling},
$\EE{}{\sum_{v\in C_{i}} w_v}\le \EE{}{|C_i|}\cdot O(\log^4 n)+\EE{}{\sum_{v\in
    C_{i-1}} w_v}\cdot O(\log^3 n)\le \beta(\log n)^{ic+5}$. 
The bound then follows from an application of the Markov bound.

The second bound follows from a more careful use of the parameters of
the two sampling lemmas, Lemmas~\ref{lem:uniformSampling}
and~\ref{lem:nonUniformSampling}. In fact, it suffices to apply these
lemmas with $\eps=e^{\Theta(1/t)}$ and $\delta=0.1$ for the
first level and $\delta=1/n^3$ for subsequent levels. 
\end{proof}

These lemmas, \ref{lem:correctness} and \ref{lem:sampleSize}, together
with the characterization theorem \ref{thm:ecDistance},
almost complete the proof of Theorem~\ref{thm:upperBound}. It remains
to bound the run time of the resulting estimation algorithm, which we
do in the next section.

\subsection{Near-Linear Time Algorithm}

We now discuss the time complexity of the algorithm, and show that the
Algorithm~\ref{alg:estimation} (Estimation) may be implemented in $n\cdot (\log
n)^{O(h)}$ time. We note that as currently described in
Algorithm~\ref{alg:estimation}, our reconstruction technique takes
time $\tilde O(hQ_b \cdot n)$ time, where $Q_b=\beta (\log n)^{O(\log_b
  n)}$ is the sample complexity upper bound from
Lemma~\ref{lem:sampleSize} (note that, combined with the algorithm
of~\cite{LMS98}, this already gives a $n^{4/3+o(1)}$ time algorithm). The
main issue is the computation of the
quantities $\delta_{v'}$, as, naively, it requires to iterate over all $k\in
[n]$.

To reduce the time complexity of the Algorithm~\ref{alg:estimation},
we define the following quantity, which replaces the quantity
$\delta_{v'}$ in the description of the algorithm:
$$
\delta'_{v'}=\min_{k=e^{i/\log
    n}: i\in[\log n\cdot \ln (3n/\beta)]}
\left(|k|+\min_{k':|k'|\le k}\tau(v',z+jl_{i+1}+k')\right).
$$

\begin{lemma}
If we use $\delta'_{v'}$ instead of $\delta_{v'}$ in
Algorithm~\ref{alg:estimation}, the new algorithm outputs at most a
$1+o(1)$ factor higher value than the original algorithm.
\end{lemma}
\begin{proof}
First we note that it is sufficient to consider only $k\in
[-3n/\beta,3n/\beta]$, since, if the algorithm uses some $k$ with
$|k|>3n/\beta$, then the resulting output is guaranteed to be
$>3n/\beta$. Also, the estimate may only increase if one restricts the
set of possible $k$'s.

Second, if we consider $k$'s that are integer powers of $e^{1/\log
  n}$, we increase the estimate by only a factor $e^{1/\log n}$. 
Over
$h=O(\log_b n)$ levels, this factor accumulates to only $e^{h/\log
  n}\le 1+o(1)$.
\end{proof}

Finally, we mention that computing all $\delta'_{v'}$ may be performed
in $O(\log^2 n)$ time after we perform the following (standard)
precomputation on the values $\tau(v',z')$ for $z'\in[n]$ and $v'\in C_{i+1}$. 
For each dyadic interval $I$, compute
$\min_{z\in I} \tau(v,z)$. Then, for each (not necessarily dyadic)
interval $I'\subset[n]$, computing $\min_{z'\in I'}\tau(v',z')$ may be
done in $O(\log n)$ time. Hence, since we consider only $O(\log
n)$ values of $k$, we obtain $O(\log^2 n)$ time per computation of
$\delta'_{v'}$.

Total running time becomes $O(h Q_b\cdot n\cdot \log^2
n)=n\cdot (\log n)^{O(\log_b n)}$.

A more technical issue that we swept under the carpet is that
distribution $\W$ defined in Lemma~\ref{lem:nonUniformSampling} is a
continuous distribution on $[1,n^3]$. However this is not an issue since a
$n^{-\Omega(1)}$ discretization suffices to obtain the same result,
with only $O(\log n)$ loss in time complexity.

\section{Query Complexity Lower Bound}\label{sec:full_lb}

We now give a full proof of our lower bound, Theorem \ref{thm:LB}. After some preliminaries, this section contains three rather technical parts: tools for analyzing indistinguishability, tools for analyzing edit distance behavior, and a finally a part where we put together all elements of the proof. 
The precise and most general forms of our lower bound appear in that final part
as Theorem~\ref{thm:main_lb} and Theorem~\ref{thm:lbMorePrecise}.

\subsection{Preliminaries}

We assume throughout that $|\Sigma|\ge 2$.
Let $x$ and $y$ be two strings. 
Define $\edd(x,y)$ to be the minimum number of character
insertions and deletions needed to transform $x$ into $y$.
Character substitution are not allowed, in contrast to $\ed(x,y)$, 
but a substitution can be simulated by a deletion followed by an insertion, 
and thus $\ed(x,y) \le \edd(x,y) \le 2\ed(x,y)$.
Observe that 
\begin{equation}  \label{eq:edLCS}
 \edd(x,y) = |x|+|y| - 2\lcs(x,y),
\end{equation}
where $\lcs(x,y)$ is the length of the longest common subsequence 
of $x$ and $y$.

\paragraph{Alignments.}
For two strings $x,y$ of length $n$,
an {\em alignment} is a function $A:[n]\to[n]\cup \{\bot\}$
that is monotonically increasing on $A^{-1}([n])$
and satisfies $x[i]=y[A(i)]$ for all $i\in A^{-1}([n])$.
Observe that an alignment between $x$ and $y$ corresponds exactly
to a common subsequence to $x$ and $y$.

\paragraph{Projection.}
For a string $x\in \Sigma^n$ and $Q \subseteq [n]$, 
we write $x|_Q$ for the string that is the projection of $x$ 
on the coordinates in $Q$. 
Clearly, $x|_Q\in \Sigma^{|Q|}$.
Similarly, if $\D$ is a probability distribution over strings in $\Sigma^n$,
we write $\D|_Q$ for the distribution that is the projection of $\D$ 
on the coordinates in $Q$. 
Clearly, $\D|_Q$ is a distribution over strings in $\Sigma^{|Q|}$. 

\paragraph{Substitution Product.}
Suppose that we have a ``mother'' string $x\in\Sigma^n$ and 
a mapping $B:\Sigma\to (\Sigma')^{n'}$ of the original alphabet 
into strings of length $n'$ over a new alphabet $\Sigma'$.
Define the \emph{substitution product} of $x$ and $B$,
denoted $x\sprod B$, to be the concatenation of $B(x_1),\cdots,B(x_n)$.
Letting $B_a=B(a)$ for each $a\in\Sigma$ 
(i.e., $B$ defines a collection of $|\Sigma|$ strings), we have
$$ x\sprod B \eqdef B_{x_1}B_{x_2}\cdots B_{x_n} \in (\Sigma')^{nn'}.$$
Similarly, for each $a \in \Sigma$, let $\D_a$ be a probability 
distribution over strings in $(\Sigma')^{n'}$.
The \emph{substitution product} of $x$ and $\D\eqdef(\D_a)_{a\in\Sigma}$, denoted $x\sprod\D$,
is defined as the probability distribution over strings in $(\Sigma')^{nn'}$ 
produced by replacing every symbol $x_i$, $1 \le i \le n$, in $x$ by an independent 
sample $B_i$ from $\D_{x_i}$.

Finally, let $\E$ be a ``mother'' probability distribution over strings
in $\Sigma^n$, and for each $a\in\Sigma$, let $\D_a$ be a probability 
distribution over strings in $(\Sigma')^{n'}$.
The \emph{substitution product} of $\E$ 
and $\D\eqdef(\D_a)_{a\in\Sigma}$, denoted $\E\sprod\D$,
is defined as the probability distribution over strings in $(\Sigma')^{nn'}$ 
produced as follows: first sample a string $x\sim \E$, 
then independently for each $i\in[n]$ sample $B_i\sim\D_{x_i}$, 
and report the concatenation $B_1B_2\ldots B_n$.

\paragraph{Shift.}
For $x \in \Sigma^n$ and integer $r$, 
let $S^r(x)$ denote a cyclic shift of $x$ (i.e.\ rotating $x$) to the left
by $r$ positions. Clearly, $S^r(x)\in\Sigma^n$.
Similarly, let $\S_s(x)$ the distribution over strings in $\Sigma^n$
produced by rotating $x$ by a random offset in $[s]$,
i.e.\ choose $r\in[s]$ uniformly at random and take $S^r(x)$.

For integers $i,j$, define $i+_n j$ to be the unique $z\in [n]$
such that $z=i+j \pmod n$.
For a set $Q$ of integers, let $Q+_n j=\{i+_n j:\ i\in Q\}$.

\begin{fact}
Let $x\in \Sigma^n$ and $Q\subset[n]$.
For every integer $r$, we have $S^r(x)|_Q = x|_{Q+_n r}$.
Thus, for every integer $s$, the probability distribution $\S_s(x)|_Q$ 
is identical to $x|_{Q+_n r}$ for a random $r\in[s]$.
\end{fact}


\subsection{Tools for Analyzing Indistinguishability}\label{sec:tools_distinguish}

In this section, we introduce tools for analyzing indistinguishability of distributions we construct.
We introduce a notion of uniform similarity, show what it implies for query complexity, give quantitative bounds on it for random cyclic shifts of random strings, and show how it composes under the substitution product.

\subsubsection{Similarity of Distributions}

We first define an auxiliary notion of similarity. Informally, a set of distributions on the same set are similar if the probability of every element in their support is the same up to a small multiplicative factor.

\begin{definition}
Let $\mathcal D_1$, \ldots, $\mathcal D_k$ be probability distributions on a finite set $\Omega$.
Let $p_i:\Omega \to [0,1]$, $1 \le i \le k$, be the probability mass function for $\mathcal D_i$.
We say that the distributions are \emph{$\alpha$-similar} if for every $\omega \in \Omega$,
$$ (1-\alpha) \cdot \max_{i=1,\ldots,k}p_i(\omega) \le \min_{i=1,\ldots,k}p_i(\omega).$$
\end{definition}

We now define uniform similarity for distributions on strings. Uniform similarity captures how the similarity between distributions on strings changes as a function of the number of queries.

\begin{definition}
Let $\mathcal D_1$, \ldots, $\mathcal D_k$ be probability distributions on $\Sigma^n$.
We say that they are \emph{uniformly $\alpha$-similar} if for every subset $Q$ of $[n]$,
the distributions $\mathcal D_1|_Q$, \ldots, $\mathcal D_k|_Q$ are $\alpha|Q|$-similar.
\end{definition}

Finally, we show that if two distributions on strings are uniformly similar, then an algorithm distinguishing strings drawn from them has to make many queries.

\begin{lemma}\label{lemma:SimilarityAlgorithms}
Let $\mathcal D_0$ and $\mathcal D_1$ be uniformly $\mu$-similar distributions on $\Sigma^n$. Let $\mathcal A$ be a randomized algorithm that makes $q$ (adaptive) queries to symbols of a string selected according to either $\mathcal D_0$ or $\mathcal D_1$, and outputs either {\tt 0} or {\tt 1}. Let $p_j$, for $j\in\zo$, be the probability that $\mathcal A$ outputs {\tt j} when the input is selected according to $\mathcal D_j$. 
Then $$\min \{p_0,p_1\} \le \frac{1 + \mu q}{2}.$$
\end{lemma}

\begin{proof}
Once the random bits of $\mathcal A$ are fixed, $\mathcal A$ can be seen as a decision tree with depth $q$ the following properties. Every internal node corresponds to a query to a specific position in the input string.
Every internal node has $|\Sigma|$ children, and the $|\Sigma|$ edges outgoing to the children are labelled with distinct symbols from $\Sigma$. Each leaf is labelled with either {\tt 0} or {\tt 1}; this is the algorithm's output, i.e. the computation ends up in a leaf if and only if the sequence of queries on the path from the root to the leaf gives the sequence described by the edge labels on the path. 

Fix for now $\mathcal A$'s random bits. Let $t$ be the probability that $\mathcal A$ outputs {\tt 0} when the input is chosen from $\mathcal D_0$, and let $t'$ be defined similarly for $\mathcal D_1$. We now show an upper bound on $t-t'$. $t$ is the probability that the computation ends up in a leaf $v$ labelled {\tt 0} for an input chosen according to $\mathcal D_0$. Consider a specific leaf $v$ labelled with {\tt 0}. The probability of ending up in the leaf equals the probability of obtaining a specific sequence of symbols for a specific sequence of at most $q$ queries. Let $t_v$ be this probability when the input is selected according to $\mathcal D_0$. The same probability for $\mathcal D_1$ must be at least $(1-q\mu)t_v$, due to the uniform $\mu$-similarity of the distributions. By summing over all leaves $v$ labelled with {\tt 0}, we have $t' \ge (1-\mu q)t$, and therefore, $t-t' \le q\mu \cdot t \le q\mu$. 

Note that $p_0$ is the expectation of $t$ over the choice of $\mathcal A$'s random bits. Analogously, $1-p_1$ is the expectation of $t'$. Since $t - t'$ is always at most $\mu q$, we have $p_0 - (1-p_1) \le \mu q$. This implies that $p_0 + p_1 \le 1 + \mu q$, and $\min\{p_0,p_1\} \le \frac{1+\mu q}{2}$.
\end{proof}

\subsubsection{Random Shifts}

In this section, we give quantitative bounds on uniform similarity between distributions created by random cyclic shifts of random strings.

Making a query into a cyclic shift of a string is equivalent to
querying the original string in a position that is shifted,
and thus, it is important to understanding what happens to 
a fixed set of $q$ queries that undergoes different shifts.
Our first lemma shows that a sufficiently large set of shifts of $q$ queries 
can be partitioned into at most $q^2$ large sets,
such that no two shifts in the same set intersect 
(in the sense that they query the same position).

\begin{lemma}\label{lemma:coloring}
Let $Q$ be a subset of $[n]$ of size $q$,
and let $Q_i\eqdef Q+_n i$ be its shift by $i$ modulo $n$.
Every $\mathcal I \subset [n]$ of size $t \ge 16q^4\ln q$
admits a $q^2$-coloring $C:\mathcal I \to [q^2]$ 
with the following two properties:
\begin{itemize}
\compactify 
 \item For all $i \ne j$ with $Q_i \cap Q_j \ne \emptyset$, 
we have $C(i) \ne C(j)$.
 \item For all $i \in [n]$, we have $|C^{-1}(i)| \ge  n /  (2q^4)$.
\end{itemize}
\end{lemma}

\begin{proof}
Let $x \in [n]$. There are exactly $q$ different indices $i$ such that $x \in Q_i$. For every $Q_i$ such that $x \in Q_i$, $x$ is an image of a different $y \in Q$ after a cyclic shift. Therefore, each $Q_i$ can intersect with at most $q(q-1)$ other sets $Q_j$.

Consider the following probabilistic construction of $C$. For
consecutive $i \in \mathcal I$, we set $C(i)$ to be a random color in
$[q^2]$ among those that were not yet assigned to sets $Q_j$ that
intersect $Q_i$. Each color $c \in [q^2]$ is considered at least
$t/q^2$ times: each time $c$ is selected it makes $c$ not be considered for at most $q(q-1)$ other $i \in \mathcal I$. Each time $c$ is considered, it is selected with probability at least $1/q^2$. 
By the Chernoff bound, the probability that a given color is selected less than $t/(2q^4)$ times is less than
$$\exp\left(-\frac{t}{q^4} \cdot \frac{1}{2^2} \cdot \frac{1}{2}\right) \le \frac{1}{q^2}.$$
By the union bound, the probability of selecting the required coloring is greater than zero, so it exists.
\end{proof}

\begin{fact}\label{fact:binomial_sum}
Let $n$ and $k$ be integers such that $1 \le k \le n$. Then
$\sum_{i=1}^{k}\binom{n}{i} \le n^k$.
\end{fact}

The following lemma shows that random shifts of random strings are likely to result in uniformly similar distributions.

\begin{lemma}\label{lem:RotationSimilarity}
Let $n \in \mathbb Z_+$ be greater than 1. Let $k \le n$ be a positive integer. Let $x_i$, $1 \le i \le k$, be uniformly and independently selected strings in $\Sigma^n$, where $2 \le |\Sigma| \le n$.
With probability $2/3$ over the selection of $x_i$'s, 
the distributions $\S_s(x_1)$, \ldots, $\S_s(x_k)$ are uniformly $\tfrac1A$-similar,
for $A \eqdef \max\left\{\log_{|\Sigma|} \sqrt[6]{\frac{s}{400 \ln n}},1\right\}$.
\end{lemma}

\begin{proof}
Let $p_{i,Q,\omega}$ be the probability of selecting a sequence $\omega \in \Sigma^{|Q|}$ from the distribution $S_s(x_i)|_Q$, where $Q \subseteq [n]$ and $1 \le i \le k$. We have to prove that with probability at least $2/3$ over the choice of $x_i$'s, it holds that for every $Q\subseteq Q$ and every $\omega \in \Sigma^{|Q|}$,
$$(1-|Q|/A) \cdot \max_{i=1,\ldots,k} p_{i,Q,\omega} \le \min_{i=1,\ldots,k} p_{i,Q,\omega}.$$

The above inequality always holds when $Q$ is empty or has at least $A$ elements.
Let $Q \subseteq [n]$ be any choice of queries, where $0 < |Q| < A$. By Fact~\ref{fact:binomial_sum}, there are at most $n^A$ such different choices of queries.
Let $q \eqdef |Q|$.
Note that $8q^4 \ln q \le 8 q^5 \le 8A^5 \le 8|\Sigma|^{5A} \le 8 \cdot \frac{s}{400 \ln n} \le s$.
This implies that we can apply Lemma~\ref{lemma:coloring}, which yields the following.
We can partition all $s$ shifts of $Q$ over $x_i$ that contribute to the distribution $S_s(x_i)|_Q$ into $q^2$ sets $\sigma_j$ such that the shifts in each of the sets are disjoint, and each of the sets has size at least $s/(2q^4)$. For each of the sets $\sigma_j$, and for each $\omega \in \Sigma^q$, the probability that fewer than $(1-\frac{q}{2A})|\sigma_j|/|\Sigma|^q$ shifts give $\omega$ is bounded by 
\begin{eqnarray*}
\exp\left(-\frac{1}{2} \cdot \left(\frac{q}{2A}\right)^2 \cdot  \frac{|\sigma_j|}{|\Sigma|^q}\right) &\le& 
\exp\left(-\frac{s}{16q^2A^2|\Sigma|^q}\right)\\
&\le& \exp\left(-\frac{s}{16A^4|\Sigma|^A}\right)\\
&\le& \exp\left(-\frac{s}{16|\Sigma|^{5A}}\right)\\
&\le& \exp\left(-\frac{1}{16} \cdot\sqrt[6]{s \cdot (400 \ln n)^5} \right)\\
&\le& \exp\left(-9.2 \sqrt[6]{s (\ln n)^5} \right),
\end{eqnarray*}
where the first bound follows from the Chernoff bound. Analogously, the probability that more than 
$(1+\frac{q}{2A})|\sigma_j|/|\Sigma|^q$ shifts give $\omega$ is bounded by 

\begin{eqnarray*}
\exp\left(-\frac{1}{4} \cdot \left(\frac{q}{2A}\right)^2 \cdot  \frac{|\sigma_j|}{|\Sigma|^q}\right) &\le& 
\exp\left(-\frac{s}{32q^2A^2|\Sigma|^q}\right)\\
&\le& \exp\left(-\frac{s}{32A^4|\Sigma|^A}\right)\\
&\le& \exp\left(-\frac{s}{32|\Sigma|^{5A}}\right)\\
&\le& \exp\left(-\frac{1}{32} \cdot\sqrt[6]{s \cdot (400 \ln n)^5} \right)\\
&\le& \exp\left(-4.6 \sqrt[6]{s (\ln n)^5} \right),
\end{eqnarray*}
where the first inequality follows from the version of the Chernoff bound that uses the fact that
$\frac{q}{2A} \le \frac{1}{2} \le 2e-1$.

We now apply the union bound to all $x_i$, all choices of $Q\subseteq[n]$ with $|Q| < A$, all corresponding sets $\sigma_j$, and all settings of $\omega\in\Sigma^{|Q|}$ to bound the probability that
$p_{i,Q,\omega}$ does not lie between $|\Sigma|^{-|Q|}\cdot(1-\frac{q}{2A})$
and $|\Sigma|^{-|Q|}\cdot(1+\frac{q}{2A})$. Assuming that $A > 1$ (otherwise, the lemma holds trivially), note first that
\begin{eqnarray*}
n \cdot n^{A} \cdot A^2 \cdot |\Sigma|^A 
& \le & n^{5A}\\
& \le & \exp\left(5A \ln n\right) \\
& \le & \exp\left(5|\Sigma|^A \ln n\right)\\
& \le & \exp\left(5\sqrt[6]{\frac{s(\ln n)^5}{400}} \right)\\
& \le & \exp\left(2.4 \cdot \sqrt[6]{s(\ln n)^5} \right).
\end{eqnarray*}
Our bound is 
$$\exp\left(2.4 \cdot \sqrt[6]{s(\ln n)^5} \right) \cdot \left(\exp\left(-9.2 \cdot\sqrt[6]{s (\ln n)^5} \right) + \exp\left(-4.6 \cdot\sqrt[6]{s (\ln n)^5} \right)\right)$$
$${}\le \exp\left(-6.8 \cdot\sqrt[6]{s (\ln n)^5} \right) + \exp\left(-2.2 \cdot\sqrt[6]{s (\ln n)^5}\right) \le 0.01 + 0.2 \le 1/3.$$
Therefore, all $p_{i,Q,\omega}$ of interest lie in the desired range with probability at least $2/3$.
Then, we know that for any $Q$ of size less than $A$, and any $\omega \in \Sigma^{|Q|}$,
we have 
\begin{eqnarray*}
\left(1-\frac{|Q|}{A} \right) \cdot \max_{i=1,\ldots,k} p_{i,Q,\omega} &\le&  \cdot
\left(1-\frac{|Q|}{A}\right) \cdot\left(1+\frac{|Q|}{2A}\right) \cdot |\Sigma|^{-|Q|}\\
&=& \left(1-\frac{|Q|}{2A}-\frac{|Q|^2}{2A^2}\right) \cdot  |\Sigma|^{-|Q|}\\
&\le& \left(1-\frac{|Q|}{2A}\right) \cdot |\Sigma|^{-|Q|}\\ \\
&\le& \min_{i=1,\ldots,k} p_{i,Q,\omega}.
\end{eqnarray*}
This implies that $\S_s(x_1)$, \ldots, $\S_s(x_k)$ are uniformly $\tfrac1A$-similar with probability at least $2/3$.
\end{proof}

\subsubsection{Amplification of Uniform Similarity via Substitution Product}

One of the key parts of our proof is the following lemma that shows that the substitution product of uniformly similar distributions amplifies uniform similarity. 

\begin{lemma}\label{lemma:SimilarityMultiplication}
Let $\D_a$ for $a \in \Sigma$, be uniformly $\alpha$-similar distributions on $(\Sigma')^{n'}$.
Let $\D \eqdef (\D_a)_{a \in \Sigma}$. Let $\mathcal E_1$, \ldots, $\mathcal E_k$ be uniformly $\beta$-similar probability distributions on $\Sigma^{n}$, for some $\beta \in [0,1]$.
Then the $k$ distributions $(\mathcal E_1 \sprod \mathcal D)$, \ldots, $(\mathcal E_{k} \sprod \mathcal D)$ are uniformly $\alpha\beta$-similar.
\end{lemma}

\begin{proof}
Fix $t,t'\in[k]$, let $X$ be a random sequence selected according to 
$\mathcal E_{t} \sprod \mathcal D$, 
and let $Y$ be a random sequence selected according to 
$\mathcal E_{t'} \sprod \mathcal D$. 
Fix a set $S \subseteq [n\cdot n']$ of indices, 
and the corresponding sequence $s$ of $|S|$ symbols from $\Sigma'$. 
To prove the lemma, it suffices to show that
\begin{equation}  \label{eq:SimilarityMultiplication}
\Pr[X|_S = s] \ge (1-\alpha\beta|S|) \cdot \Pr[Y|_S = s],
\end{equation}
since in particular the inequality holds for $t$ that minimizes $\Pr[X|_S = s]$,
and for $t'$ that maximizes $\Pr[Y|_S = s]$.

Recall that each $(\mathcal E_{j} \sprod \mathcal D)$ is generated by first selecting a string $x$ according to $\mathcal E_j$, and then concatenating $n$ blocks, where the $i$-th block is independently selected from $\mathcal D_{x_i}$. For $i \in [n]$ and $b \in \Sigma$, let $p_{i,b}$ be the probability of drawing from $\mathcal D_b$ a sequence that when used as the $i$-th block, matches $s$ on the indices in $S$ (if the block is not queried, set $p_{i,b}=1$).
Let $q_i$ be the number of indices in $S$ that belong to the $i$-th block. Since $\mathcal D_b$ for $b \in \Sigma$ are $\alpha$-similar, for every $i \in [n]$, it holds that $(1-\alpha q_i) \cdot \max_{b\in\Sigma} p_{i,b} \le \min_{b\in\Sigma} p_{i,b}$. For every $i \in [n]$, define $\alpha^\star_i \eqdef \min_{b \in \Sigma}p_{i,b}$ and $\beta^\star_i \eqdef \max_{b \in \Sigma}p_{i,b}$. 
We thus have
\begin{equation}  \label{eq:2}
(1-\alpha q_i) \beta^\star_i \le \alpha^\star_i.
\end{equation}

The following process outputs {\tt 1} with probability $\Pr[Y|_S = s]$. Whenever we say that the process outputs a value, {\tt 0} or {\tt 1}, it also terminates. First, for every block $i \in [n]$, the process independently picks a random real $r_i$ in $[0,1]$. It also independently draws a random sequence $c \in \Sigma^{n}$ according to $\mathcal E_{t'}$. If $r_i > \beta^\star_i$ for at least one $i$, the process outputs {\tt 0}. Otherwise, let $Q=\{i \in [n]:\ r_i > \alpha^\star_i\}$. If $r_i\le p_{i,c_i}$ for all $i \in Q$, the process outputs {\tt 1}. Otherwise, it outputs {\tt 0}. The correspondence between the probability of outputting {\tt 1} and $\Pr[Y|_S = s]$ directly follows from the fact that each of the random variables $r_i$ simulates selecting a sequence that matches $s$ on indices in $S$ with the right success probability, i.e., $p_{i,c_i}$, and the fact that block substitutions are independent. The important difference, which we exploit later, is that not all symbols of $c$ have always impact on whether the above process outputs {\tt 0} or {\tt 1}.

For every $Q\subseteq [n]$, let $p'_Q$ be the probability that the above process selected $Q$. Furthermore, let $p''_{Q,c}$ be the conditional probability of outputting {\tt 1}, given that the process selected a given $Q \subseteq [n]$, and a given $c \in \Sigma^n$. It holds 
$$
\Pr\left[Y|_S = s\right]
=
\sum_{Q\subseteq[n]} p'_Q \cdot \mathbb E_{c \leftarrow \mathcal E_{t'}} \left[p''_{Q,c}\right].
$$
Notice that for two different $c_1, c_2 \in \Sigma^n$, we have 
$p''_{Q,c_1}=p''_{Q,c_1}$ if $c_1|_Q = c_2|_Q$, since this probability only depends on the symbols at indices in $Q$. 
Thus, for $\tilde c \in \Sigma^{|Q|}$ we can define $\tilde p_{Q,\tilde c}$ to be equal to $p''_{Q,c}$ for any $c\in\Sigma$ such that $c|_Q = \tilde c$. We can now write
$$\Pr\left[Y|_S = s\right]
=
\sum_{Q\subseteq[n]} p'_Q \cdot \mathbb E_{\tilde c \leftarrow \mathcal E_{t'}|_Q}
\left[\tilde p_{Q,\tilde c}\right],
$$
and analogously,
$$\Pr\left[X|_S = s\right]
=
\sum_{Q\subseteq[n]} p'_Q \cdot \mathbb E_{\tilde c \leftarrow \mathcal E_{t}|_Q}
\left[\tilde p_{Q,\tilde c}\right].
$$
Due to the uniform $\beta$-similarity of $\mathcal E_{t'}$ and $\mathcal E_{t}$,
we know that for every $Q\subset[n]$, the probability of selecting each $\tilde c \in \Sigma^{|Q|}$ from $\mathcal E_{t}|_Q$ is at least $(1-\beta|Q|)$ times the probability of selecting the same $\tilde c$ from $\mathcal E_{t'}|_Q$. This implies that
$$\mathbb E_{\tilde c \leftarrow \mathcal E_{t}|_Q}
\left[\tilde p_{Q,\tilde c}\right] 
\ge (1-\beta|Q|)\cdot \mathbb E_{\tilde c \leftarrow \mathcal E_{t'}|_Q}
\left[\tilde p_{Q,\tilde c}\right].$$
We obtain
\begin{eqnarray}
\nonumber\Pr\left[Y|_S = s\right] - \Pr\left[X|_S = s\right] &=&
\sum_{Q\subseteq[n]} p'_Q \cdot \left(
\mathbb E_{\tilde c \leftarrow \mathcal E_{t'}|_Q} \left[\tilde p_{Q,\tilde c}\right]
-
\mathbb E_{\tilde c \leftarrow \mathcal E_{t}|_Q} \left[\tilde p_{Q,\tilde c}\right]
\right).
\\
\nonumber&\le & \sum_{Q\subseteq[n]} p'_Q \cdot \beta|Q| \cdot \mathbb E_{\tilde c \leftarrow \mathcal E_{t'}|_Q} \left[\tilde p_{Q,\tilde c}\right]\\
\nonumber&=& \beta \cdot \sum_{Q\subseteq[n]} p'_Q \cdot |Q| \cdot \mathbb E_{c \leftarrow \mathcal E_{t'}} \left[p''_{Q,c}\right]\\
\label{eq:bound_diff}
&=&\beta \cdot 
\mathbb E_{c \leftarrow \mathcal E_{t'}}
\left[\sum_{Q\subseteq[n]} p'_Q \cdot p''_{Q,c} \cdot |Q|\right].
\end{eqnarray}

Fix now any $c \in \Sigma^n$ for which the process outputs {\tt 1} with positive probability. 
The expected size of $Q$ for the fixed $c$, given that the process outputs {\tt 1}, can be written as
$$
  \EX\left[|Q|\ \Big|\Big.\ \mbox{process outputs {\tt 1}}\right]
  = \frac{\sum_{Q\subseteq[n]} p'_Q \cdot p''_{Q,c} \cdot |Q|}{\sum_{Q\subseteq[n]} p'_Q \cdot p''_{Q,c} }$$
The probability that a given $i\in[n]$ belongs to $Q$ for the fixed $c$, given that the process outputs {\tt 1} 
equals $\frac{p_{i,c_i}-\alpha^\star_i}{p_{i,c_i}}$. This follows from the two facts (a) if the process outputs {\tt 1} then $r_i$ is uniformly distributed on $[0,p_{i,c_i}]$; and (b) $i \in Q$ if and only if $r_i \in (\alpha^\star_i,\beta^\star_i]$. We have 
\begin{equation}  \label{eq:3}
\frac{p_{i,c_i}-\alpha^\star_i}{p_{i,c_i}} \le \frac{\beta^\star_i-\alpha^\star_i}{\beta^\star_i}
\le \frac{\alpha q_i \cdot \beta^\star_i}{\beta^\star_i} = \alpha q_i.
\end{equation}
By the linearity of expectation, the expected size of $Q$ in this setting is at most $\sum_{i\in[n]}\alpha q_i = \alpha\cdot|S|$. Therefore, 
\begin{equation}\label{eq:exp_size_Q}
\sum_{Q\subseteq[n]} p'_Q \cdot p''_{Q,c} \cdot |Q| \le 
\alpha\cdot|S| \cdot \sum_{Q\subseteq[n]} p'_Q \cdot p''_{Q,c}.
\end{equation}
Note that the inequality trivially holds also for $c$ for which the process always outputs {\tt 0}; both sides of the inequality equal 0.

By plugging (\ref{eq:exp_size_Q}) into (\ref{eq:bound_diff}), we obtain
\begin{eqnarray*}
\Pr\left[Y|_S = s\right] - \Pr\left[X|_S = s\right] &\le&
\beta \cdot \mathbb E_{c \leftarrow \mathcal E_{t'}}
\left[\alpha\cdot|S| \cdot \sum_{Q\subseteq[n]} p'_Q \cdot p''_{Q,c}\right]\\
&=&
\alpha\beta\cdot|S| \cdot \mathbb E_{c \leftarrow \mathcal E_{t'}}
\left[\sum_{Q\subseteq[n]} p'_Q \cdot p''_{Q,c}\right]\\
&=& \alpha\beta\cdot|S| \cdot \Pr\left[Y|_S = s\right].
\end{eqnarray*}
This proves \eqref{eq:SimilarityMultiplication} 
and completes the proof of the lemma.
\end{proof}

\subsection{Tools for Analyzing Edit Distance}\label{sec:tools_edit_distance}

This section provides tools to analyze how the edit distance
changes under a under substitution product.
We present two separate results with different guarantees, 
one is more useful for a large alphabet, the other for a small alphabet. 
The latter is used in the final step of reduction to binary alphabet.

\subsubsection{Distance between random strings}

The next bound is well-known, see also \cite{CS75,BGNS99,Lueker09}.
We reproduce it here for completeness.

\begin{lemma} \label{lem:RandomStrings}
Let $x,y\in\Sigma^n$ be chosen uniformly at random.
Then 
$$ \Pr\left[\lcs(x,y) \geq 5n/\sqrt{|\Sigma|}\right] 
  \leq e^{-5n/\sqrt{|\Sigma|}}.$$
\end{lemma}
\begin{proof}
Let $c\eqdef 5 > e^{1.5}$ and $t\eqdef cn/\sqrt{|\Sigma|}$.
The number of potential alignments of size $t$ between two strings 
of length $n$ is at most $\binom{n}{t}^2\le (\tfrac{ne}{t})^{2t}$.
Each of them indeed becomes an alignment of $x,y$ 
(i.e.\ symbols that are supposed to align are equal) 
with probability at most $1/|\Sigma|^{t}$.
Applying a union bound,
\begin{align*}
  \Pr[\lcs(x,y) \ge t]
  & \le (\tfrac{ne}{t})^{2t} / |\Sigma|^{t} 
    \le ({e^2c^{-2}|\Sigma|})^t \cdot |\Sigma|^{-t}
    \le e^{-t}. 
\qedhere\end{align*}
\end{proof}

\subsubsection{Distance under substitution product (large alphabet)}

We proceed to analyze how the edit distance between two strings, say $\ed(x,y)$,
changes when we perform a substitution product, i.e.\ $\ed(x\sprod B,y\sprod B)$.
The bounds we obtain are additive, and are thus most effective 
when the edit distance $\ed(x,y)$ is large (linear in the strings length).
Furthermore, they depend on $\lambda_B\in[0,1]$, which denotes the maximum 
normalized LCS between distinct images of $B:\Sigma\to(\Sigma')^{n'}$,
hence they are most effective when $\lambda_B$ is small, 
essentially requiring a large alphabet $\Sigma'$.

\begin{theorem}\label{thm:sprodAdditive}
Let $x,y\in \Sigma^n$ and $B:\Sigma\to(\Sigma')^{n'}$. Then 
$$ n'\cdot \edd(x,y) -8nn' \sqrt{\lambda_B}
   \ \le\ \edd(x\sprod B, y\sprod B) 
   \ \le\ n'\cdot \edd(x,y),$$
where 
$\lambda_B\eqdef 
  \max \Big\{\tfrac{\lcs(B(a),B(b))}{n'} :\ a\neq b\in\Sigma\Big\}$.
\end{theorem}

Before proving the theorem, we state a corollary that will turn to be most 
useful.
The corollary follows from Theorem \ref{thm:sprodAdditive} by letting $\Sigma'=\Sigma$, 
and using Lemma \ref{lem:RandomStrings} together with 
a union bound over all pairs $B(a),B(b)$ (while assuming $n'\ge |\Sigma|$).

\begin{corollary} \label{cor:sprodAdditiveRandom}
Assume $|\Sigma|\ge 2$ and $n'\ge |\Sigma|$ is sufficiently large 
(i.e.\ at least some absolute constant $c'$).
Let $B:\Sigma\to(\Sigma)^{n'}$ be a random function,
i.e.\ for each $a\in\Sigma$ choose $B(a)$ uniformly at random.
Then with probability at least $1-2^{-n'/|\Sigma|}$,
for all $n$ and all $x,y\in\Sigma^n$,
$$ 0
  \ \le\ n'\cdot \edd(x,y) - \edd(x\sprod B, y\sprod B) 
  \ \le\ O(nn'/|\Sigma|^{1/4}).$$
\end{corollary}

\begin{proof}[Proof of Theorem \ref{thm:sprodAdditive}]
By using the direct connection \eqref{eq:edLCS} between $\edd(x,y)$ and $\lcs(x,y)$,
it clearly suffices to prove 
\begin{equation}  \label{eq:sprodAdditiveLCS}
  n'\cdot \lcs(x,y) 
  \le \lcs(x\sprod B, y\sprod B) 
  \le n'\cdot \lcs(x,y) + 4nn'\sqrt{\lambda_B}.
\end{equation}
Throughout, we assume the natural partitioning of $x,y$ 
into $n$ blocks of length $n'$.

The first inequality above is immediate.
Indeed, give an (optimal) alignment between $x$ and $y$, do the following;
for each $(i,j)$ such that $x_i$ is aligned with $y_j$,
align the entire $i$-th block in $x\sprod B$ with the entire $j$-th block 
in $y\sprod B$.
It is easily verified that the result is indeed an alignment and has 
size $n'\cdot \edd(x,y)$.

To prove the second inequality above, 
fix an optimal alignment $A$ between $x\sprod B$ and $y\sprod B$;
we shall construct an $\hat A$ alignment for $x,y$ in three stages,
namely, first pruning $A$ into $A'$, then pruning it further into $A''$,
and finally constructing $\hat A$.
Define the \emph{span} of a block $b$ in either $x\sprod B$ or $y\sprod B$
(under the current alignment) to be the number of blocks in the other string
to which it is aligned in at least one position
(e.g. the span of block $i$ in $x\sprod B$ is the number of blocks $j$
for which at least one position $p$ in block $i$ satisfies that $A(p)$ 
is in block $j$.)

Now iterate the following step:
``unalign'' a block (in either $x\sprod B$ or $y\sprod B$) completely 
whenever its span is greater than $s\eqdef 2/\sqrt{\lambda_B}$.
Let $A'$ be the resulting alignment; its size is $|A'| \ge |A| - 4nn'/s$
because each iteration is triggered by a distinct block, 
the total span of all these blocks is at most $4n$,
hence the total number of iterations is at most $4n/s$. 

Next, iterate the following step (starting with $A'$ as the current alignment): 
remove alignments between two blocks (one in $x\sprod B$ and one in $y\sprod B$)
if, in one of the two blocks, 
at most $\lambda_Bn'$ positions are aligned to the other block.
Let $A''$ be the resulting alignment; its size is $|A''| \ge |A'| - ns\cdot \lambda_Bn'$
because each iteration is triggered by a distinct pair of blocks,
out of at most $ns$ pairs (by the span bound above).

This alignment $A''$ has size $|A''| \ge |A| -4nn'/s -nn's\lambda_B$.
Furthermore, if between two blocks, 
say block $i$ in $x\sprod B$ and block $j$ in $y\sprod B$,
the number of aligned positions is at least one,
then this number is actually greater than $\lambda_Bn'$ (by construction of $A''$)
and thus $x[i]=y[j]$ (by definition of $\lambda_Bn'$).

Finally, construct an alignment $\hat A$ between $x$ and $y$,
where initially, $\hat A(i)=\bot$ for all $i\in [n]$.
Think of the alignment $A''$ as the set of aligned positions,
namely $\{(p,q)\in[n]\times[n]:\ A''(p)=q\}$.
Let $\blk_{x\sprod B}(p)$ denote the number of the block in $x\sprod B$ 
which contains $p$, and similarly for positions $q$ in $y\sprod B$.
Now scan $A''$, as a set of pairs, in lexicographic order.
More specifically, initialize $(p,q)$ to be the first edge in $A''$,
and iterate the following step: 
assign $\hat A(\blk_{x\sprod B}(p))=\blk_{y\sprod B}(q)$, 
and advance $(p,q)$ according to the lexicographic order 
so that both coordinates now belong to new blocks,
i.e.\ set it to be the next pair $(p',q')\in A''$
for which both $\blk_{x\sprod B}(p')>\blk_{x\sprod B}(p)$
and $\blk_{y\sprod B}(q')>\blk_{y\sprod B}(q)$.
We claim that $\hat A$ is an alignment between $x$ and $y$.
To see this, consider the moment when we assign some $\hat A(i)=j$.
Then the corresponding blocks in $x\sprod B$ and $y\sprod B$
contain at least one pair of positions that are aligned under $A''$, 
and thus, as argued above, $x[i]=y[j]$.
In addition, all subsequent assignments of the form $\hat A(i')=j'$ 
satisfy that both $i'>i$ and $j'>j$.
Hence $\hat A$ is indeed an alignment.

En route to bounding the size of $\hat A$, we claim that each iteration scans 
(i.e.\ advances the current pair by) at most $n'$ pairs from $A''$.
To see this, consider an iteration where we assign some $\hat A(i)=j$.
Every pair $(p,q)\in A''$ that is scanned in this iteration
satisfies that either $i=\blk_{x\sprod B}(p)$ or $j=\blk_{x\sprod B}(p)$.
Each of these two requirements can be satisfied by at most $n'$ pairs,
and together at most $2n'$ pairs are scanned.
By the fact that $A''$ is monotone, it can be easily verified that
at least one of the two requirements must be satisfied by all scanned pairs,
hence the total number of scanned pairs is at most $n'$.

Using the claim, we get that $|\hat A| \le |A''|/n'$
(recall that each iteration also makes one assignment to $\hat A$).
It immediately follows that 
\[
  n'\cdot \lcs(x,y) 
  \ge n'\cdot |\hat A|
  \ge |A''|
  \ge |A| -4nn'/s -nn's\lambda_B
  = \lcs(x\sprod B,y\sprod B)-4nn'\sqrt{\lambda_B},
\]
which completes the proof of \eqref{eq:sprodAdditiveLCS} and of 
Theorem \ref{thm:sprodAdditive}.
\end{proof}

\subsubsection{Distance under substitution product (any alphabet)}

We give another analysis for how the edit distance between two strings, 
say $\ed(x,y)$,
changes when we perform a substitution product, i.e.\ $\ed(x\sprod B,y\sprod B)$.
The bounds we obtain here are multiplicative, 
and may be used as a final step of alphabet reduction 
(say, from a large alphabet to the binary one).

\begin{theorem}\label{thm:sprodFar}
Let $B:\Sigma\to(\Sigma')^{n'}$, and suppose that 
(i) for every $a\neq b\in\Sigma$, we have 
\[
  \lcs(B_a,B_b) \le \tfrac{15}{16}n';
\]
and (ii) for 
every $a,b,c\in\Sigma$ (possibly equal), 
and every substring $B'$ of (the concatenation) $B_bB_c$
that has length $n'$ and overlaps each of $B_b$ and $B_c$ by at least $n'/10$,
we have 
\[ \lcs(B_a,B') \le 0.98n'. \]
Then for all $x,y\in \Sigma^n$,
\begin{equation}  \label{eq:sprodFar}
  c_1 n'\cdot \edd(x,y)
   \le \edd(x\sprod B, y\sprod B) 
   \le n'\cdot \edd(x,y),
\end{equation}
where $0<c_1<1$ is an absolute constant.
\end{theorem}

Before proving the theorem, let us show that it is applicable for 
a random mapping $B$, 
by proving two extensions of Lemma \ref{lem:RandomStrings}.
Unlike the latter, the lemmas below are effective also for small alphabet size.
\begin{lemma}\label{lem:randomFar}
Suppose $|\Sigma|\ge 2$ and let $x,y\in \Sigma^n$ be chosen uniformly at random.
Then with probability at least $1-|\Sigma|^{-l/8}$, the following holds:
for every substring $x'$ in $x$ of length $l\ge 24$,
and every length $l$ substring $y'$ in $B_b$, 
we have 
$$ \lcs(x',y') \le \tfrac{15}{16}l.$$
\end{lemma}
\begin{proof}
Set $\alpha\eqdef 1/16$.
Fix $l$ and the positions of $x'$ inside $x$ and of $y'$ inside $y$.
Then $x'$ and $y'$ are chosen at random from $\Sigma^l$, hence
\[
  \Pr[\lcs(x',y') \ge (1-\alpha) l]
  \le \textstyle\binom{l}{(1-\alpha)l}^2 |\Sigma|^{-(1-\alpha)l}
  \le (\tfrac{e}{\alpha})^{2\alpha l} |\Sigma|^{-(1-\alpha)l}
  \le |\Sigma|^{-l/4},
\]
where the last inequality uses $|\Sigma|\ge2$.

Now apply a union bound over all possible positions of $x'$ and $y'$
and all values of $l$.
It follows that the probability that $x$ and $y$ contain length $l$ substrings 
$x'$ and $y'$ (respectively) with $\lcs(x',y')\ge (1-\alpha)l$ is at most 
$|\Sigma|^3 \cdot|\Sigma|^{-l/4}  \leq |\Sigma'|^{-l/8}$,
if only $l$ is sufficiently large.
\end{proof}

The next lemma is an easy consequence of Lemma \ref{lem:randomFar}.
It follows by applying a union bound and observing that disjoint
substrings of $B(a)$ are independent.

\begin{lemma}\label{lem:randomFar2}
Let $B:\Sigma\to(\Sigma')^{n'}$ be chosen uniformly at random
for $|\Sigma'|\ge 2$ and $n' \ge 1000\log|\Sigma|$.
Then with probability at least $1-|\Sigma'|^{-\Omega(n')}$, 
$B$ satisfies the properties (i) and (ii) described in 
Theorem \ref{thm:sprodFar}.
\end{lemma}

\begin{proof}[Proof of Theorem \ref{thm:sprodFar}]
The last inequality in \eqref{eq:sprodFar} is straightforward.
Indeed, whenever $x_i$ is aligned against $y_j$,
we have $x_i=y_j$ and $B(x_i)=B(y_j)$,
hence we can align the corresponding blocks in $x\sprod B$ and $y\sprod B$.
We immediately get that $\lcs(x\sprod B, y\sprod B) \ge n'\cdot \lcs(x,y)$.

Let us now prove the first inequality.
Denote $R\eqdef\edd(x\sprod B, y\sprod B)$,
and fix a corresponding alignment between the two strings.
The string $x\sprod B$ is naturally partitioned into $n$ blocks of length $n'$.
The total number of coordinates in $x\sprod B$ that are unaligned 
(to $y\sprod B$) is exactly $R/2$, which is $R/2n$ in an average block.

We now prune this alignment in two steps.
First, ``unaliagn'' each block in $x\sprod B$ with at least
$(nn'/100R)\cdot (R/2n) = n'/200$ unaligned coordinates.
By averaging (or Markov's inequality), this step applies to
at most $100R/nn'$-fraction of the $n$ blocks.

Next, define the \emph{gap} of a block in $x\sprod B$ to be the difference
(in the positions) between the first and last positions in $y\sprod B$
that are aligned against a coordinate in $x\sprod B$.
The second pruning step is to unalign every block in $x\sprod B$ 
whose gap is at least $1.01n'$.
Every such block can be identified with a set of at least $n'/100$ 
unaligned positions in $y\sprod B$ (sandwiched inside the gap), 
hence these sets (for different blocks) are all disjoint,
and the number of such blocks is at most $(R/2)/(n'/100)=50R/n'$.

Now consider one of the remaining blocks (at least $n-100R/n'-50R/n'$ blocks).
By our pruning, for each such block $i$ we can find a corresponding 
substring of length $n'$ in $y\sprod B$ with at least 
$n'-n'/200-n'/100 > 0.98n'$ aligned pairs (between these two substrings).
Using the property (ii) of $B$, the corresponding substring in $y\sprod B$
must have overlap of at least $0.9n'$ with some block of $y\sprod B$ 
(recall that $y\sprod B$ is also naturally partitioned into length $n'$ blocks).
Thus, for each such block $i$ in $x\sprod B$ there is a corresponding
block $j$ in $y\sprod B$, such that these two blocks contain at least
$0.9n'-0.02n'=0.88n'$ aligned pairs.
By the property (i) of $B$, it follows that the corresponding
coordinates in $x$ and in $y$ are equal, i.e. $x_i=y_j$.
Observe that distinct blocks $i$ in $x\sprod B$ are matched in this way 
to distinct blocks $j$ in $y\sprod B$ 
(because the initial substrings in $y\sprod B$ were non-overlapping,
and they each more than $n'/2$ overlap with a distinct block $j$).

It is easily verified that the above process gives an alignment 
between $x$ and $y$.
Recall that the number of coordinates in $x$ that are not aligned 
in this process is at most $150R/n'$, hence
$\edd(x,y) \le 300R/n'$, and this completes the proof.
\end{proof}

\subsection{The Lower Bound}

We now put all the elements of our proof together. We start by describing hard distributions, and then prove their properties. We also give a slightly more precise version of the lower bound for polynomial approximation factors in a separate subsection.

\subsubsection{The Construction of Hard Distributions}\label{sect:hard_distro}

We give a probabilistic construction for the hard distributions. We have two basic parameters, $n$ which is roughly the length of strings, and $\alpha$ which is the approximation factor. We require that $2<\alpha \ll n/\log n$. 
The strings length is actually smaller than $n$ (for $n$ large enough), 
but our query complexity lower bound hold also for length $n$,
e.g., by a simple argument of padding by a fixed string.

We now define the hard distributions.
\begin{enumerate}
\item Fix an alphabet $\Sigma$ of size $\lceil 5^2 \cdot 2^{16}\cdot \log_{\alpha}^4 n\rceil$.
\item Set:
\begin{itemize}
 \item $T\eqdef \left\lceil 1000 \cdot \log |\Sigma| \right\rceil$.
 \item $\beta \eqdef \begin{cases}
          \alpha, &\hbox{if $\alpha < n^{1/3}$,}\\
          \frac{n}{\alpha \ln n}, & \hbox{otherwise}.\\
         \end{cases}$
 \item $s \eqdef \left\lceil 400\beta \ln n \cdot |\Sigma|^{12} \right\rceil$, thus $s = O(\beta \cdot \log n \cdot \log_\alpha^{48} n)$.
 \item $B \eqdef \left\lceil 8\alpha s \cdot \log_\alpha n \right\rceil$, implying that $B = O(\alpha \beta\log n \cdot \log_\alpha^{49} n)$. Notice that $B < \frac{n}{T}$ for $n$ large enough. If $\alpha < n^{1/3}$, then $B = \tilde O(n^{2/3})$. Otherwise, $\log_\alpha n \le 3$, $\log |\Sigma| = O(1)$, and $B = o(n)$. 
\end{itemize}

\item Select at random $|\Sigma|$ strings of length $B$, 
denoted $x_a$ for $a \in \Sigma$.

\item
Define $|\Sigma|$ corresponding distributions $\D_a$. 
For each $a \in \Sigma$, let
$$\D_a \eqdef \S_s(x_a),$$
and set 
$$\D \eqdef (\D_a)_{a \in \Sigma}.$$
\item Define by induction on $i$a a collection of distributions $\E_{i,a}$ for $a\in\Sigma$.
As the base case, set 
$$\E_{1,a} \eqdef \D_a.$$
For $i > 1$, set
$$\E_{i,a} \eqdef \E_{i-1,a} \sprod \D.$$

\item Let $i_\star \eqdef \left\lfloor \log_{B} \frac{n}{T}\right\rfloor$. Note that the distributions $\E_{i_\star,a}$ are defined on strings of length $B^{i_\star}$, which is
is of course at most $\frac{n}{T}$, but due to an earlier observation,
we also know that $i_\star \ge 1$, for $n$ large enough.

\item Fix distinct $a_\star,b_\star\in\Sigma$. Let $\mathcal F_0 \eqdef \E_{i_\star,a_\star}$ and $\mathcal F_1 \eqdef \E_{i_\star,b_\star}$.

\item Pick a random mapping $R : \Sigma \to \zo^{T}$. Let
$\mathcal F'_0 \eqdef \mathcal F_0 \sprod R$
and 
$\mathcal F'_1 \eqdef \mathcal F_1 \sprod R$.
Note that the strings drawn from $\mathcal F'_0$ and $\mathcal F'_1$ are of length at most $n$.
\end{enumerate}

Notice the construction is probabilistic only because of step 
\#3 (the base strings $x_a$),
and \#8 (the randomized reduction to binary alphabet).

\subsubsection{Proof of the Query Complexity Lower Bound}\label{sect:final_lb}

The next theorem shows that:
\begin{itemize}
\compactify
 \item Every two strings selected from the same distribution $\mathcal F_i$ 
are always close in edit distance.

 \item With non-zero probability (recall the construction is probabilistic), 
distribution $\mathcal F_0$ produces strings that are far, in edit distance, from strings produced by $\mathcal F_1$, yet distinguishing between these cases requires many queries.
\end{itemize}
Essentially the same properties hold also for $\mathcal F'_0$ and $\mathcal F'_1$.

\begin{theorem}\label{thm:main_lb}
Consider a randomized algorithm that is given full access to a string in $\Sigma^n$, and query access to another string in $\Sigma^n$. Let $2<\alpha\le o(n/\log n)$.
If the algorithm distinguishes, with probability at least $2/3$, edit distance $\ge n/2$ from $\le n/(4\alpha)$, then it makes
$$\left(2+\Omega\left(\frac{\log \alpha}{\log\log n}\right)\right)^{\max\left\{1,\Omega\left(\frac{\log n}{\log \alpha + \log\log n}\right)\right\}}$$
queries for $\alpha < n^{1/3}$, and $\Omega\left(\log \frac{n}{\alpha \ln n}\right)$ queries
for $\alpha \ge n^{1/3}$. The bound holds even for $|\Sigma| = O(\log_\alpha^4 n)$.

For $\Sigma = \zo$, the same number of queries is required to distinguish edit distance 
$\ge c_1n/2$ and $\le c_1n/(4\alpha)$, where $c_1 \in (0,1)$ is the constant from Theorem~\ref{thm:sprodFar}.
\end{theorem}

\begin{proof}
We use the construction described in Section~\ref{sect:hard_distro}.
Recall that $i_\star \ge 1$, for $n$ large enough, and that $i_\star \le \log_B n$.

Let $F:\Sigma\to \Sigma^B$ be defined as 
$F(a) \eqdef x_a$
for every $a \in \Sigma$. We define $y_{i,a}$ inductively. Let $y_{1,a} \eqdef x_a$ for every $a \in \Sigma$, then for $i > 1$ define $y_{i,a} \eqdef y_{i-1,a} \sprod F$.

We now claim that for every word $z$ with non-zero probability in $\E_{i,a}$ for $a \in \Sigma$, we have 
$$\frac{\ed(z,y_{i,a})}{B^i} \le \frac{i \cdot 2 \cdot s}{B} \le \frac{i}{4\alpha \log_\alpha n}.$$
This follows by induction on $i$, since every rotation by $s$ can be ``reversed'' with at most
$s$ insertions and $s$ deletions. In particular,
$$\frac{\ed(z,y_{i_\star,a})}{B^{i_\star}} \le \frac{\log_B n}{4\alpha \log_\alpha n} = 
\frac{\log \alpha}{4\alpha \log B} \le \frac{1}{4\alpha},$$
where the last inequality is because $\alpha \le B$.

It follows from Lemma~\ref{lem:RandomStrings} and the union bound that with probability 
$$1-|\Sigma|^2\cdot e^{-5B/\sqrt{|\Sigma|}} \ge 1-|\Sigma|^2 \cdot e^{-5|\Sigma|} \ge
1-e^{-3|\Sigma|} \ge 1-e^{-3} \ge 2/3$$
(over the choice of $F$, i.e. $x_a$ for $a\in\Sigma$), 
that for all $a\neq b\in\Sigma$ we have $\lcs(x_a,x_b) \le 5B/ \sqrt{|\Sigma|}$, that
is, the value corresponding to $\sqrt{\lambda_B}$ in Lemma~\ref{thm:sprodAdditive}
is at most $\sqrt{5/\sqrt{|\Sigma|}} \le 1/(16\log_\alpha n)$. 
We assume henceforth this event occurs.
Then by Lemma~\ref{thm:sprodAdditive} and induction, 
we have that for all $a\neq b$, 
$$\edd(y_{i,a},y_{i,b})\ge B^{i} \left(2 - \frac{i}{2\log_\alpha n}\right)$$
which gives
\begin{eqnarray*}
\ed(y_{i_\star,a_\star},y_{i_\star,b_\star}) &\ge&
\frac{1}{2}\edd(y_{i_\star,a_\star},y_{i_\star,b_\star}) \ge B^{i_\star} \left(1 - \frac{i_\star}{4\log_\alpha n}\right) \ge B^{i_\star} \left(1 - \frac{\log \alpha}{4\log B}\right)\\
&\ge&
B^{i_\star} \left(1 - \frac{1}{4}\right) = \frac{3}{4}B^{i_\star}.
\end{eqnarray*}

Consider now an algorithm that is given full access to the string $y_{i_\star,a_\star}$ and query access to some other string $z$.
If $z$ comes from $\mathcal F_0 = \E_{i_\star,a_\star}$, then $\ed(y_{i_\star,a_\star},z) \le \frac{B^{i_\star}}{4\alpha}$.
If $z$ comes from $\mathcal F_1 = \E_{i_\star,b_\star}$, then 
$\ed(y_{i_\star,a_\star},z) \ge \frac{3}{4} B^{i_\star} - \frac{1}{4\alpha} B^{i_\star} \ge \frac{1}{2} B^{i_\star}$
by the triangle inequality.

We now show that the algorithm has to make many queries to learn whether $z$ is drawn from $\mathcal F_0$ or from $\mathcal F_1$. By Lemma~\ref{lem:RotationSimilarity}, with probability at least $2/3$ over the choice of $x_a$'s, $\E_{1,a}$'s are uniformly $\tfrac1A$-similar, for
\begin{eqnarray*}
A &\eqdef& \log_{|\Sigma|}\sqrt[6]{\frac{s}{400\ln B}} \ge \log_{|\Sigma|}\sqrt[6]{\beta \cdot |\Sigma|^{12}}
 = 2 + \frac{\log \beta}{6\log|\Sigma|}. 
\end{eqnarray*}
Note that both the above statement regarding $\tfrac1A$-similarity as well as the earlier requirement
that $\lcs(x_a,x_b)$ be small for all $a \ne b$, are satisfied with non-zero probability.

Observe that $\log |\Sigma| = \Theta(1 + \log (\frac{\log n}{\log \alpha}))$.
For $\alpha < n^{1/3}$,
$$A = 2 + \Omega\left(\frac{\log \alpha}{1+\log\left(\frac{\log n}{\log \alpha}\right)}\right)
= 2 + \Omega\left(\frac{\log \alpha}{\log\log n}\right).$$
For $\alpha \ge n^{1/3}$,
$$A \ge 2 + \Omega\left(\frac{\log \frac{n}{\alpha \ln n}}{1+\log\left(\frac{\log n}{\log \alpha}\right)}\right) 
  \ge \Omega\left(\log \frac{n}{\alpha \ln n}\right),$$
where the last transition follows since
$\frac{\log n}{\log \alpha}=\Theta(1)$ and $\alpha = o(n/\log n)$.

By using Lemma~\ref{lemma:SimilarityMultiplication} over $\E_{i,a}$'s, we have that
$\E_{i,a}$'s are uniformly $\tfrac{1}{A^i}$-similar. It now follows from Lemma~\ref{lemma:SimilarityAlgorithms} that an algorithm that distinguishes whether its input $z$ is drawn from $\mathcal F_0 = \E_{i_\star,a_\star}$ or from $\mathcal F_1 = \E_{i_\star,b_\star}$ with probability at least $2/3$, must make at least $A^{i_\star}/3$ queries to $z$.
Consider first the case of $\alpha < n^{1/3}$.
We have $i_\star = \Omega\left(\frac{\log n}{\log B}\right) = \Omega\left(\frac{\log n}{\log \alpha + \log\log n}\right)$. The number of queries we obtain is
$$\left(2+\Omega\left(\frac{\log \alpha}{\log\log n}\right)\right)^{\max\left\{1,\Omega\left(\frac{\log n}{\log \alpha + \log\log n}\right)\right\}}.$$
For $\alpha \ge n^{1/3}$ we have $i_\star \ge 1$, and the algorithm must make 
$\Omega\left(\log \frac{n}{\alpha \ln n}\right)$ queries. This finishes the prove of the first part of the theorem, which states a lower bound for an alphabet of size $\Theta(\log_\alpha^4 n)$.

For the second part of the theorem regarding alphabet $\Sigma=\zo$, we use the distributions from the first part, but we employ the mapping $\mathcal R:\Sigma \to \zo^{T}$ to replace every symbol in $\Sigma$ with a binary string of length $T$. Lemma~\ref{lem:randomFar2} and Theorem~\ref{thm:sprodFar} state that if $R$ is chosen at random, then with non-zero probability, $R$ preserves (normalized) edit distance up to a multiplicative $c_1$. Using such a mapping $R$ and $\alpha/c_1$ instead of $\alpha$ in the entire proof, we obtain the desired gap in edit distance between $\mathcal F_0'$ and $\mathcal F_1'$.
The number of required queries remains the same after the mapping, because every symbol in a string obtained
from $\mathcal F'_0$ or $\mathcal F'_1$ is a function of a single symbol from a string obtained from
$\mathcal F_0$ or $\mathcal F_1$, respectively. An algorithm using few queries to distinguish 
$\mathcal F_0'$ from $\mathcal F_1'$ would therefore imply an algorithm with similar query complexity to distinguish $\mathcal F_0$ from $\mathcal F_1$, which is not possible.
\end{proof}

\subsubsection{A More Precise Lower Bound for Polynomial Approximation Factors}

We now state a more precise statement that specifies the exponent for polynomial approximation factors.

\begin{theorem}
\label{thm:lbMorePrecise}
Let $\lambda$ be a fixed constant in $(0,1)$. Let $t$ be the largest positive integer
such that $\lambda \cdot t < 1$.

Consider an algorithm that is given a string in $\Sigma^n$, and query access to another string in $\Sigma^n$. 
If the algorithm correctly distinguishes edit distance $\ge n/2$ and $\le n/(4n^{\lambda})$
with probability at least $2/3$, then it needs $\Omega(\log^{t} n)$ queries,
even for $|\Sigma| = O(1)$.

For $\Sigma = \zo$, the same number of queries is required to distinguish edit distance 
$\ge c_1n/2$ and $\le c_1n/(4n^{\lambda})$, where $c_1 \in (0,1)$ is the constant from Theorem~\ref{thm:sprodFar}.
\end{theorem}

\begin{proof}
The proof is a modification of the proof of Lemma~\ref{thm:main_lb}. 
We reuse the same construction with the following differences:
\begin{itemize}
 \item We set $\alpha \eqdef n^{\lambda}$. This is our approximation factor.
 \item We set $\beta \eqdef n^{\frac{1}{2} \left(\frac{1}{t} - \lambda \right)}$. This is up to a logarithmic factor the shift at every level of recursion
\end{itemize}
$T$, $s$, $B$, $|\Sigma|$ are defined in the same way as functions of $\alpha$ and $\beta$.
Note that $B = \Theta\left(n^{\frac{1}{2} \left(\frac{1}{t} + \lambda \right)} \log n\right)$ and $T = \Theta(1)$.
This implies that for sufficiently large $n$, $i_\star = \lfloor \log_B \frac{n}{T} \rfloor = t$, because
$B^t = \tilde\Theta\left(n^{\frac{1+\lambda t}{2}}\right) = o(n)$, and
$B^{t+1} = \tilde\Theta\left(n^{\frac{1}{2}+\frac{1}{2t}+\frac{\lambda(t+1)}{2}}\right) = \tilde\Omega\left(n^{1+\frac{1}{2t}}\right)
= \omega(n)$.

As in the proof of Lemma~\ref{thm:main_lb}, we achieve the desired separation in edit distance. Recall that the number of
queries an algorithm must make is $\Omega(A^{i_\star})$, where
$$A \ge 2 + \frac{\log \beta}{6\log|\Sigma|} = \Omega(\log n).$$
Thus, the number of required queries equals $\Omega(\log^t n)$.
\end{proof}

\small
\bibliography{bibfile,more} 
\bibliographystyle{alpha}

\end{document}